\documentclass[11pt]{article}
\usepackage{amssymb} 
\usepackage{amsmath}
\usepackage{amsthm}   
\usepackage{epsfig}
\usepackage{subfigure}
\usepackage{a4wide}
\usepackage{setspace}
\usepackage{rotating} 
\usepackage{color} 
\usepackage{geometry}  
\usepackage{ifthen}   
\newcommand{\user}{alex}      

\newfont{\smcal}{cmu10 scaled 1200}
\newfont{\handw}{cmmi10 scaled 1200}
\newfont{\handws}{cmmi10 scaled 800}
\newtheorem{Prop}{Proposition}[section]
\newtheorem{Lem}[Prop]{Lemma}

\newtheorem{Th}[Prop]{Theorem}
\newtheorem{Rm}[Prop]{Remark}
\newtheorem{Rms}[Prop]{Remarks}
\newtheorem{Def}[Prop]{Definition}

\newtheorem{Ex}[Prop]{Example}
\newtheorem{As}[Prop]{Assumption}

\newcommand{\grad}{\mbox{\rm grad}}

\newcommand{\var}{\mbox{\rm Var}}

\newcommand{\supp}{\mbox{\rm supp}}

\newcommand{\dint}{\int\!\!\!\!\int}

\newcommand{\sint}{{\int_0^1}}

\newcommand{\ba}{\ensuremath{{\bf a}}}%
\newcommand{\bb}{\ensuremath{{\bf b}}}%

\newcommand{\wA}{\ensuremath{\widetilde{A}}}%
\newcommand{\wB}{\ensuremath{\widetilde{B}}}%
\newcommand{\E}{\ensuremath{\mathbb{E}}}%
\newcommand{\T}{\ensuremath{\mathbb{T}}}%
\newcommand{\Z}{\ensuremath{\mathbb{Z}}}%
\newcommand{\R}{\ensuremath{\mathbb{R}}}%
\newcommand{\C}{\ensuremath{\mathbb{C}}}
\newcommand{\calD}{\ensuremath{\mathcal{D}}}%

\newcommand{\vart}{\vartheta}
\newcommand{\hvart}{\hat\vartheta}
\newcommand{\varep}{\epsilon}

\newcommand{\Hess}{\,\mbox{\rm Hess}}
\newcommand{\IM}{\,\mbox{\rm Im}}
\newcommand{\RE}{\,\mbox{\rm Re}}

\newcommand{\tM}{\widetilde{M}}%
\newcommand{\tF}{\widetilde{F}}%

\newcommand{\konvD}{\overset{\cal D}{\rightarrow}} 

\newcommand{\teacher}[1]{%
	\ifthenelse{	\isundefined{\showme}}%
		{}%
		{\fbox{\parbox{0.9\textwidth}{#1 \begin{center}Dieser Abschnitt erscheint nur in der \emph{teacher's version}.\end{center}}}}}

\begin{document}
      \title{Drift Estimation in Sparse Sequential Dynamic Imaging: with Application to Nanoscale Fluorescence Microscopy}
        \author{Alexander Hartmann\footnote{Institute for Mathematical Stochastics, Georg-August-University G\"ottingen, Germany}, Stephan Huckemann$^{*, }$\footnote{Felix Bernstein Institute for Mathematical Statistics in the Biosciences, Georg-August-University G\"ottingen, Germany}, J\"orn Dannemann$^{*}$,\\ Oskar Laitenberger\footnote{Laser Laboratory, G\"ottingen, Germany}, Claudia Geisler$^{\ddagger}$, Alexander Egner$^{\ddagger}$, and Axel Munk$^{*, \dagger, }$\footnote{Max-Planck-Institute for biophysical Chemistry, G\"ottingen, Germany} $^{, }$\footnote{Corresponding author. munk@math.uni-goettingen.de}}
        \maketitle

\onehalfspacing

\begin{abstract}
A major challenge in many modern superresolution fluorescence microscopy techniques at the nanoscale lies in the correct alignment of long sequences of sparse but spatially and temporally highly resolved images. This is caused by the temporal drift of the protein structure, e.g. due to temporal thermal inhomogeneity of the object of interest or its supporting area during the observation process.
We develop a simple semiparametric model for drift correction in SMS microscopy. Then we propose an M-estimator for the drift and show its asymptotic normality. This is used to correct the final image and it is shown that this purely statistical method is competitive with state of the art calibration techniques which require to incorporate fiducial markers into the specimen. Moreover, a simple bootstrap algorithm allows to quantify the precision of the drift estimate and its effect on the final image estimation. We argue that purely statistical drift correction is even more robust than fiducial tracking rendering the latter superfluous in many applications. The practicability of our method is demonstrated by a simulation study and by an SMS application.
This serves as a prototype for many other typical imaging techniques where sparse observations with highly temporal resolution are blurred by motion of the object to be reconstructed.
\end{abstract}

\par
\vspace{9pt}
\par
\vspace{9pt}
\noindent {\it AMS 2000 Subject Classification:} \begin{minipage}[t]{10cm}
Primary 62M10, 62M40, Secondary 92C05, 92C37
 \end{minipage}
\par

\section{Introduction}
Optical fluorescence imaging is an important tool in the life sciences for studying biological molecules at subcellular level.
Until 20 years ago the Abb\'e diffraction barrier stood for more than hundred years as a physical limitation of spatial resolution for any kind of light microscopy. This amounts to a resolution level of about 250 nm (approx. half the wave length of visible light) in lateral and 500 nm in axial direction. The diffraction barrier is attributable to the fact that two features that are closer than the resolution level cannot be distinguished in a light micrograph because they merge into one another. Meanwhile, this barrier has been overcome by imaging features that are within such a diffraction limited area not simultaneously but consecutively by changing their ability to generate contrast in time \cite{Hell2009b}. In the case of fluorescence microscopy, this means changing the fluorophore's ability to send out a fluorescence photon or to change the properties of the emitted fluorescence photon, e.g. its color. This switching has been implemented by several techniques \cite{Hell2003, Betzig2006, Rust2006, 
Hess2006} which has initiated a revolution in cell imaging. Nowadays, biological molecules can even be viewed ``at work'' at a resolution level down to 10 -- 20 nm which gives entirely new insights into the signalling and transport processes within cells (see e.g. \cite{Westphal2008, Berning2012, Jones2011, Huang2013}, to mention a few). State of the art nanoscale microscopy can be roughly divided into two distinct categories: In the targeted mode (ensemble based), the fluorophores (markers) are switched at a known (precisely defined) coordinate, whereas in the stochastic mode, the fluorophores are switched at random (initially unknown) locations. The first includes techniques such as stimulated emission depletion (STED) \cite{Hell1994, Klar2000, Schmidt2008}, saturated patterned excitation microscopy (SPEM) \cite{Heintzmann2002} or saturated structured illumination microscopy (SSIM) \cite{Gustafsson2005m}, and reversible saturable optical fluorescence transitions (RESOLFT) \cite{Hofmann2005, Hell2003}. 
Due to the direct targeting the acquisition time of these techniques is usually relatively short and the sample drift is not a major source of blurring.

In contrast, in its stochastic switching (or single marker switching, SMS) mode, fluorescence microscopy is performed by illuminating the whole sample but with a low switching light intensity, assuring that with high probability only a few (random) markers are in their fluorescent state at any time. These techniques are being developed rapidly during the last years and they include \emph{stochastic optical reconstruction microscopy (STORM)} \cite{Rust2006, Holden2011}, \emph{photoactivated localization microscopy (PALM)} \cite{Betzig2006}, \emph{fluorescence photoactivation localization microscopy (FPALM)} \cite{Hess2006}, and \emph{PALM with independently running acquisition (PALMIRA)} \cite{Geisler2007, Egner2007} and \cite{Hell2007} for a survey.

Given the fact that in SMS microscopy a sufficient number of marker molecules has to be imaged in order to generate a representative view of the sample, SMS experiments provide a huge number (e.g. in the range of several tens of thousands) of highly time resolved images (frames), each of which contains very little but sparse information. In this setting, methods have recently been developed which make explicit use of this sparseness for image reconstruction, e.g. using a sparsity enforcing penalty or prior, see \cite{Babcock2012, Cox2012, Holden2011, Zhu2012, Quan2011, Hafi2014}. The unknown marker positions are usually determined by calculating the centroid of their observed diffraction patterns which renders more sophisticated deconvolution methods unnecessary. Obviously, this physically enforces spatial sparseness and the localization accuracy can be ${\sqrt{N}}$ times better than the initial resolution of the microscope, where $N$ is the average number of detected photons within the individual 
diffraction patterns \cite{Thompson2002}. The markers localized within each frame are 
then registered in highly time and space resolved position histograms (see Figure \ref{single.images:fig}), the overlay of which represents the final SMS-image.

A major motivation for this paper is, however, that the measurement process in SMS microscopy typically takes several minutes. Hence, the image is blurred if the object drifts over significant distances during this time. This drift may be caused by temperature variations (thermal drift) during the measurement process and external systematic movements of the optical device (mechanical drift). As can be seen in Figure \ref{data.reconstr:fig} (left upper display) this drift is the major source of blurring. The issue of correcting for the, per se unknown, motion of the object in the sparse position histograms is well known and it is therefore current practice to incorporate fiducial markers (e.g. bright fluorescent microspheres) into the sample in order to register subsequent frames. This is technically demanding and expensive. Often the fiducials also outshine relevant parts of the image, hence it would be an important achievement to develop methods which allow to estimate the drift of the sample without incorporation of fiducials into the sample.

A first attempt has been made in \cite{Geisler2012}, who suggested a heuristic correlation method to align subsequent frames properly (see \cite{Deschout2014} for a recent survey on this issue). In this paper we will treat this problem in a statistically rigorous way. We argue that a parametric model for the drift function is often appropriate and we suggest an M-estimator for it. See the right hand side of Figure \ref{data.reconstr:fig} for the image of the recordings of a $\beta$-tubulin network (network I) within a mammalian cell, which was obtained after drift correction with our M-estimator, to be developed in section \ref{drift-est-general:sec}. We will show the asymptotic normality of this estimator as the acquisition time increases, and we argue that this is the ``right asymptotics'' in SMS microscopy due to relatively long acquisition times which inherently come along with this technique. From this asymptotics we obtain simple bootstrap confidence bands for the drift function and finally improved 
estimates of the image together with a measure to access the statistical uncertainty of the aligned images.

We stress that our asymptotics is substantially different to that underlying many other image alignment and registration methods where at each time step data from the full image is observed and hence asymptotics concerns the number of pixels tending to infinity.

Finally, our method is compared in real world applications from SMS microscopy with calibration using a fiducial marker. We show that our method is at least as competitive revealing the incorporation of fiducials as not necessary in the analysis and processing of SMS images.

\begin{figure}[h!]
\ifthenelse{\equal{\user}{alex}}{
  \includegraphics[width=1\textwidth]{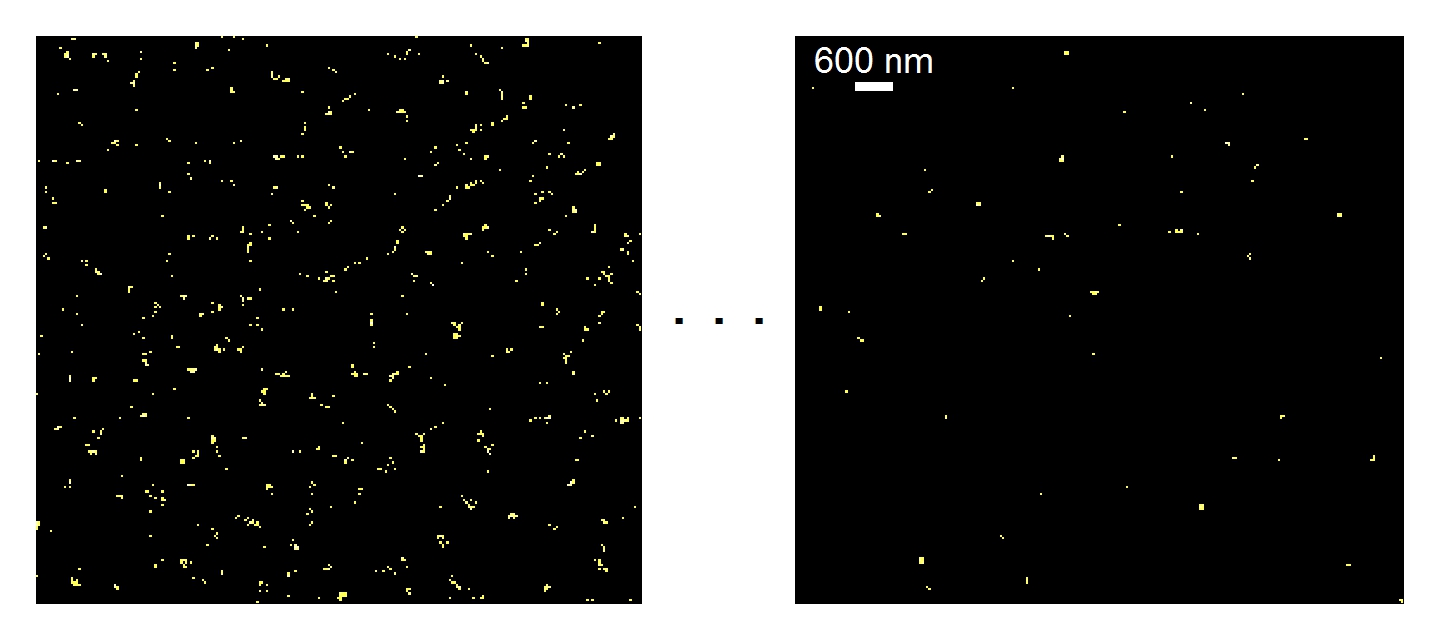}
  }{}
\ifthenelse{\equal{\user}{stephan}}{
  \includegraphics[width=1\textwidth]{../IMG/Cage552_PVA_YG200_03_Einzelbilder.jpg}
  }{}
 \caption{\it Superimposed position histograms derived from the first (left) and last (right) 20 frames of an SMS experiment. The sample (network I) is a fixed Vero-cell Abberior Cage 552-labelled $\beta$-tubulin network. The superimposed position histograms of all 40,000 frames of the experiment and a reconstruction derived from them are shown in Figure \ref{data.reconstr:fig}.}
 \label{single.images:fig}
\end{figure}

\begin{figure}[b!]
\centering
\ifthenelse{\equal{\user}{alex}}{
  \includegraphics[width=0.75\textwidth]{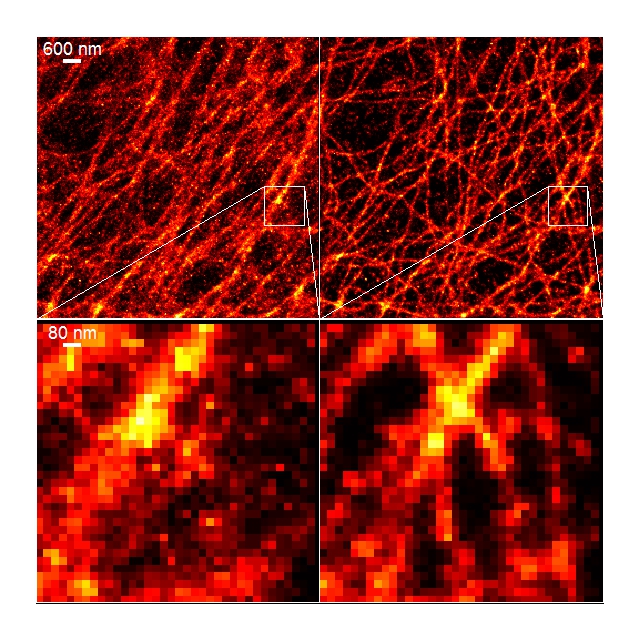}
  }{}
\ifthenelse{\equal{\user}{stephan}}{
  \includegraphics[width=0.75\textwidth]{../IMG/Cage552_PVA_YG200_03_linear-quadratic_details.jpg}
  }{}
 \caption{\it Left column: SMS acquisition of the Abberior Cage 552-labelled $\beta$-tubulin network I in a fixed Vero-cell. Top left: drift blurred position histogram. Top right: reconstructed position histogram under a linear-quadratic drift model. Bottom left: detailed view inside the white box above. Bottom right: reconstructed detail after drift correction.}
 \label{data.reconstr:fig}
\end{figure}

\paragraph{A simple drift model for SMS microscopy.}
The data acquisition process in SMS microscopy is a two step process (on-switching of marker molecules and subsequent read-out of their fluorescent signal) and we refer for details to \cite{Betzig2006, Geisler2012, Hell2009b, Hess2006}. However, as the data represents single photon counts recorded with an array of photodetectors, a spatial (thinned) Poisson process (possibly corrupted by some background noise) with unknown intensity follows from the independence assumption of photon emittance of different markers. The unknown intensity $\lambda$ of the Poisson process is linked to the unknown marker density $f$, say, by a convolution $K$ which is determined by the optical system, $\lambda = K * f$.
In ensemble based microscopy the focal spot is scanned through the sample. This requires an additional deconvolution step which can be helpful to obtain improved resolution (see \cite{Vardi1985, Silverman1990, Nowak2000, Cavalier2002, Antoniadis2006, Zhang2008, Frick2013, Bigot2013} for several Poisson deconvolution methods).
In contrast, in SMS microscopy as considered in this paper, the center of each spot already serves as a very accurate location estimate of the marker molecule because of the enforced sparsity (see \cite{Aspelmeier2014}). Therefore, we adopt current practice, and a sophisticated deconvolution step is not required.

As you may draw from Figure \ref{data.reconstr:fig}, indeed the major source of blurring in SMS microscopy comes from sample drift, rather than from optical blurring as the technique is designed to be physically sparse.

The second simplification is motivated from the fact, that since the number of photon counts in SMS experiments is usually quite high we restrict to use a (heteroscedastic) Gaussian model as an approximation to the Poisson model for large intensity $f$ in the following.
We nevertheless did some simulations for a Poisson model warranting that approximation appropriate (see Section \ref{sec:simulations}).
Hence, an approximate model for the above SMS scenario is thus 
given by (possibly after an offset correction)
	\begin{eqnarray}\label{few-observations:model}
	 Z_{j,t} &=& f\big(x_{j,t}-\delta_t\big) +\tilde\sigma_{j,t} \varep_{j,t},~~\varep_{j,t}\stackrel{i.i.d.}{\sim}{\cal N}(0,1),
	 \end{eqnarray}
with noise levels $\tilde\sigma_{j,t} > 0$. Here, for each time point $t \in \T':=\left\{0,\frac{1}{T'},\frac{2}{T'},\dotsc,\frac{T'-1}{T'}\right\}$ ($T'>0$ is the total number of frames) one observes $Z_{j, t}$ at (relatively few) locations $x_{j,t}$ which are assumed to be randomly selected from an equidistant grid of size $n=N^2$ of equally spaced pixels in the unit square $[0,1]^2$ (the image domain), $j \in J_t \subseteq \{1,\dotsc, n\}$, $n_t := \# J_t$.
The underlying unknown true marker intensity $f$ which is assumed to be square integrable is shifted by an unknown drift function $\delta_t$ and has to be estimated together with $\delta_t$.
The variances $\tilde\sigma_{j, t}^2$ model spatial and temporal inhomogeneities and are unknown, in general. In particular, in a pure Poisson model they would equal the signal $f$ itself.
As elaborated before, in the low energy stochastic switching scenario we assume that only very few (i.e. $n_t$ is small) and sufficiently distant pixel locations are selected by the switch-on process. In consequence, the errors $\varep_{j,t}$ ($t \in \T'$, $j\in J_t$) can be assumed to be independently distributed for different time points $t_1, t_2 \in \T'$, $t_1\neq t_2$, even if the affiliated pixel locations $x_{j,t_1}, x_{j,t_2}$ are identical. Actually, in model (\ref{few-observations:model}) $f(x_{j, t}-\delta_t)$ has to be rescaled with the relative amount of total intensity at time $t$, i.e. multiplied by $n_t/\sum_{s \in \T} n_s$, because the intensity of the images scales with $n_t$. This will be suppressed in the following, however, as any estimate of $f(x_{j,t})$ can be rescaled with this (observable) number.
We will see (Assumption \ref{convergent-sigma:as}) that our method does not require any assumption on $J_t$ or $n_t$ (besides of $n_t \geq 1$, which is always true for SMS microscopy as the sampling rate is never chosen below). Therefore, our results and conditions will be stated for fixed and arbitrary values $n_t$. Note that for SMS microscopy $n_t$ and $J_t$ are strictly speaking random and their exact distributions will depend on the fluorophore characteristics. Then, however, Assumption \ref{convergent-sigma:as} and our main Theorems \ref{asymptotic-normality:thm}, \ref{clt:thm} hold analogously for this situation, e.g. when the convergence in (\ref{convergent-sigma:eq}) is now a.s., which can be derived from the strong law of large numbers for non-identically distributed random variables.

In contrast to the usual asymptotics in imaging, where the pixel number $n$ tends to infinity as the discretization level increases, we have to consider here the novel scenario where the total pixel number is fixed while the number of time frames $T'$ tends to infinity.
In SMS nanoscopy, $T'$ typically ranges from 10,000 to 40,000, corresponding to a time resolution of several milliseconds.

\paragraph{The Fourier drift model.}
For the unknown image and its drift we propose a Fourier type cutoff-estimator.
Therefore, in the following it is convenient to rewrite (\ref{few-observations:model}) in terms of the spectral observations, i.e. the discrete two-dimensional Fourier transform at every time point $t\in \T$
\begin{eqnarray}\label{FT-few-observations:model}
	Y^t_k &:=& \frac1{n_t} \sum_{j \in J_t} Z_{j,t}e^{-2\pi i\langle k,x_{j,t}\rangle}~=~f^{\delta_t,t}_k + W_k^t, \quad k\in\mathbb Z^2\,
\end{eqnarray}
with suitable independent complex normal variables $W^t_k$ and $f_k^{\delta_t, t}$ the Fourier coefficients of $f^{\delta_t}(x) = f(x-\delta_t)$.
This allows to exploit the two-dimensional \emph{shift property}
	\begin{equation}
	  \label{shift:prop}
      f^{\delta_t,t}_{k} = \frac1{n_t} \sum_{j \in J_t} e^{-2\pi i \langle k , x_{j,t}\rangle} f(x_{j,t}-\delta_t)
    	= e^{-2\pi i \left\langle k , \delta_t \right\rangle} f^t_{k},\quad  f^t_{k} :=  \frac1{n_t} \sum_{j \in J_t} f(z_{j,t}) e^{-2\pi i\langle k,z_{j,t}\rangle},
	\end{equation}
	with $z_{j,t} \equiv x_{j,t} -\delta_t$ mod $[0,1]^2$.

A common proceeding in SMS microscopy (and in other imaging techniques) is binning (i.e. adding up) subsequent frames. As the total acquisition time is rather long ($T' \geq 10{,}000$), it is then reasonable to assume that we observe a big enough part of the image so that the averages over $T'$ of the Fourier coefficients $f_k^t$ are good approximations to the $f_k = \int_{[0, 1]^2} f(x) e^{-2\pi i\langle k, x\rangle} dx$\label{fkapprox}. This leads typically to $T \in \{20, \dotsc, 2000\}$ binned frames, depending on the bin width driven by the particular application. In our data application (see Figure \ref{data.reconstr:fig} and Section \ref{sec:application}), we used $T = 2000$. Therefore, in the following, we consider binned frames only.
Note that binning leaves model (\ref{few-observations:model}) qualitatively unchanged. In the following, we write $T$ and $\T$ instead of $T'$ and $\T'$, respectively, and we will denote the binned values again with $Z_{j, t}$ in (\ref{few-observations:model}) and so on.

Assuming this, the model (\ref{FT-few-observations:model}) simplifies to the following model underlying all of the theoretical considerations of this paper.
Because the $n_t$ are fixed and observable, we rewrite $\sigma_{j,t} = \tilde\sigma_{j,t}/\sqrt{n_t}$.

\begin{Def}
        For a $[0,1]^2$-periodic image $f$ the \emph{Fourier drift model} of SMS microscopy is given by
	\begin{eqnarray}\label{simple:model}
	 	Y^t_k &=&e^{-2\pi i \left\langle k , \delta_t \right\rangle} f_{k} + \,W_k^t, \quad k\in\Z^2, ~t\in \T
	\end{eqnarray}
with independent \emph{complex normals}
	\begin{eqnarray}\label{complex-Gauss-noise:def}
	  W_k^t &=& \frac1{\sqrt{n_t}} \sum_{j \in J_t} e^{-2\pi i \langle k,x_{j,t}\rangle} \sigma_{j,t}\varep_{j,t}, \quad \sigma_{j, t} > 0
	  \end{eqnarray}
	where we assume that $\varep_{j,t} \stackrel{i.i.d.}{\sim} {\cal N}(0,1)$. Then, the real and imaginary parts of $W^t_{k}$ are independently normally distributed with zero mean and $W^t_k$ is independent of $W^{t'}_{k'}$ unless $t=t'$ and $k=k'$.
\end{Def}


\paragraph{Relation to the literature.}
The asymptotics considered in section \ref{drift-est-general:sec} requires rather involved computations and notably, they are different from various approaches and asymptotics in the literature. Note that $n_t$ in model (\ref{few-observations:model}) is typically small in our setting (as it is the core in SMS microscopy) and does not tend to infinity. Hence, our approach is different from \emph{time dynamic imaging} \cite{Foroosh2002, Huang2005, Papenberg2006, Cuzol2007, Allassonniere2007, Fleet2006, Bruhn2005, Weickert2001, Li2013} where in each time step a (rather) complete sample of the entire image has to be recorded. This is in strict contrast to SMS microscopy which provides only a few markers in each time frame.
Therefore, this situation is also different from \cite{GamboaLoubesMaza2007} as well as from \cite{BigotGamboaVimond2009} although we borrow the idea of a Procrustes type estimator based on minimizing a suitable contrast functional, cf. \cite{Gow}. While the two afore mentioned recent references consider a finite number of images perturbed by Gaussian noise, each of which comes with an individual unknown similarity transform  constant over time (more specifically, translated to our setup, \cite{GamboaLoubesMaza2007} consider one-dimensional images each subject to a one-dimensional translation whereas \cite{BigotGamboaVimond2009} consider two-dimensional images each subject to a two-dimensional similarity transformation), they show that the transformations of interest can be consistently estimated with asymptotic normality when the number of pixel observations $n$ tends to infinity corresponding to an increasing signal-to-noise ratio for each image. Motivated by SMS microscopy, in our work, we have to 
swap the asymptotics as the time $T$ goes to infinity and, while not considering the full similarity group, we additionally allow for a time dependent drift. Since we consider drifts only, in contrast to \cite{BigotGamboaVimond2009}, our
method readily extends to higher dimensions, e.g. to three-dimensional images given by voxel locations.

We note that the shift property of the Fourier transform which has motivated our approach is crucial in many related methods based on FFT \cite{reddy,BigotGamboaVimond2009}.

At this point we conclude the methodological part of the introduction by noting that our work goes far beyond SMS microscopy and can be potentially used for other purposes, such as noisy object or particle tracking, when only small parts of the object are registered at each time step as it is the case for heavily undersampled magnetic resonance imaging \cite{Li2013}. Extensions to nonparametric drifts are possible and will be the topic of subsequent research. Finally, in a sense our work is complementary to the issue of testing in fluorescence microscopy whether a protein structure has significantly changed in time as in \cite{BissantzClaeskensHolzmannMunk2009}.

\paragraph{This paper is organized as follows.} In Section \ref{drift-est-general:sec} we present our main theoretical results. In particular, in Section \ref{basic-assumptions:scn} we provide for the main assumptions on the model, propose an estimator for the parameter of the drift model in Section \ref{sec:est} and derive consistency and asymptotic normality in Section \ref{main:sec} under mild assumptions. In Section \ref{sec:simulations} we illustrate the proposed method in a simulation study. Finally, we apply our method to SMS nanoscopy data  in Section \ref{sec:application} and give a detailed discussion of the results including bootstrapping of confidence regions in Section \ref{Bootstrap:scn}. Most of the proofs are deferred to the Appendix.

\paragraph{Software.} An accompanying software package \emph{R\_ImageDrift} can be found at\\
\texttt{www{.}stochastik{.}math{.}uni-goettingen{.}de/R\_ImageDrift}.

\section{Drift Estimation in a Sequence of Sparse Images: Theory}\label{drift-est-general:sec}
\subsection{Basic Assumptions for the Fourier Drift Model}\label{basic-assumptions:scn}

	In view of the Fourier methods employed we identify the two-dimensional image domain
	$[0,1]^2\subset \mathbb R^2$ with the complex unit square $D:=\{z\in \mathbb C: 0\leq \RE(z),\IM(z)\leq 1\}\subset \mathbb C$. For a point $x\in \mathbb R^2$ we have the real coordinates $(x)_1,(x)_2$ and identify $\big((x)_1,(x)_2\big)$ with $z = (x)_1 + i(x)_2$.

	For a complex vector $z\in \mathbb C^k$, $\RE(z)$ and $\IM(z)\in\mathbb R^k$ denote the corresponding real and imaginary parts and $|z|$ denotes the absolute value of $z$ or equivalently the Euclidean norm of $\mathbb R^{2k}$. For $k=(k_1,k_2)\in \Z^2$ we set
	$|k| := \max(|k_1|,|k_2|)$.
	
	For a point $x\in\mathbb R^d$, $d\in\mathbb N$,  $\|x\|$ denotes the Euclidean norm whereas $\|f\|_2$ denotes the usual norm of $f\in L^2(D)$.
	
	\begin{As}\label{image-drift:as}
	For the drift function $\delta_t: [0,1] \to D, t \mapsto \delta_t$ we assume a parametric model
	$$
	\delta_t = \delta_t^{\vart}, \vart \in \Theta
	$$
	with $\Theta\subseteq\R^d$ being a compact subset. Compactness is assumed for technical reasons. We mention that this can be relaxed with some additional effort. The parameter $\vart_0\in \Theta$ of the true shift $t\mapsto \delta_t^{\vart_0}$ is unknown. In order to avoid boundary effects, we assume that $f$ is supported on a compact subset $\calD \subset (0,1)^2$ and that $\cup_{\vart\in\Theta,t\in[0,1]} \supp f(\cdot-\delta_t^\vart) \subset \mathcal{D}$. Moreover, in order to apply the estimation method below based on Fourier transforms, we assume that $f$ is extended $1$-periodically into the two spatial directions. Also, we assume that $f$ has no period length strictly smaller than $1$.
	\end{As}

	\begin{Ex}\label{polynomial-drift:ex} Clearly, the choice of the parametric drift model is crucial for the model (\ref{simple:model}).
	As a prime example we consider the linear drift model:
	$$
	\delta_t^\vart = \delta_t^{(\alpha,\beta)} = \alpha t + \beta
	$$
	with $\vart=(\alpha,\beta)\in \Theta \subset \R^{2}\times\R^{2}$. This can easily be extended to a polynomial drift model of degree $d$:
	$$
	\delta_t^\vart = \delta_t^{(\alpha_1,\dotsc,\alpha_d,\beta)} = \beta + \alpha_1 t + \cdots + \alpha_d t^d
	$$
	with $\vart=(\alpha_1,\dotsc,\alpha_d,\beta) \in \R^{2d+2}$. In the SMS data presented in Section \ref{sec:application} the linear, quadratic and cubic models ($d \in \{1, 2, 3\}$) will be applied.
	\end{Ex}

	It is easy to see that the drift models as defined above are not identifiable per se, since the intercept $\beta$ can either be made explicit or absorbed into the function $f$.
	Because at the initial time $t=0$ we do not expect any drift we may assume $\delta_{0}(\vart)\equiv 0$. For the drift models proposed this results in the standard restriction $\beta=0$. In general, we require the identifiability of $\vart$ from the parametrized drift $\delta_t^\vart$, i.e. $\delta_t^\vart=\delta_t^{\vart_0}$ for all $t$ implies $\vart=\vart_0$. Moreover, recall that we exclude that $f$ be periodic with period length $< 1$, for otherwise $\delta_t$ were only well defined modulo the period length.

 \subsection{A Fourier Based M-Estimator}\label{sec:est}
	Recall the Fourier drift model (\ref{simple:model})
	with independent complex Gaussian noise $W_k^t$ as in (\ref{complex-Gauss-noise:def}).
	Note that
	\[ \frac{1}{T}\sum_{t \in \T} e^{2\pi i \left\langle k , \delta_t^\vart\right\rangle} {Y}^t_{k}
	  = f_k + \frac{1}{T}\sum_{t \in \T} e^{2\pi i \left\langle k , \delta_t^\vart\right\rangle} W_k^t \]
	converges a.s. to $f_k$ for $T \to \infty$, since the last term on the r{.}h{.}s{.} is the mean of independent centered Gaussian random variables which, under mild assumptions, vanishes asymptotically due to Kolmogorov's law of large numbers (see e.g. \cite[Theorem 2.3.10]{Sen1993}).
    Motivated by this observation (cf. \cite{GamboaLoubesMaza2007} for the case of a fixed $T$ and $n = n_t\to \infty$) we define the \emph{empirical contrast functional}
	\begin{eqnarray}\label{eqn:contrast}
	M_T(\vart) :=  \frac{1}{T} \sum_{|k| < \xi_T} \sum_{t\in\T}  \left|e^{2\pi i \left\langle k , \delta_t^\vart\right\rangle}\,Y^{t}_k - \frac{1}{T}\sum_{t'\in\T} e^{2\pi i \left\langle k , \delta_{t'}^\vart\right\rangle}\,Y^{t'}_k\right|^2.
	\end{eqnarray}
	The  threshold condition $|k| < \xi_T$ with $\xi_T>0$ suitably chosen will ensure
	 convergence of the right hand side of (\ref{eqn:contrast}). Our first result provides for a suitable choice of $\xi_T$.
	We note that more subtly than thresholding, one could follow \cite{GamboaLoubesMaza2007} who sum over all $\mathbb Z^2$ and use suitable spectral weight functions $\omega_T(k,t)$. Introducing the abbreviation
	$$h_k(\delta_t^\vart) := e^{2\pi i \left\langle k , \delta_t^\vart\right\rangle}\, $$
	rewrite
	\begin{eqnarray*}
	 M_T(\vart) &=&  \sum_{|k| < \xi_T} \left(
	\frac{1}{T} \sum_{t\in\T}
	\left| h_k(\delta^\vart_t)\, {Y}^t_{k} \right|^2
	- \left|\frac{1}{T} \sum_{t'\in\T}  h_k(\delta^\vart_{t'})\, {Y}^{t'}_{k}\right|^2 \right)
	~=~M_T^0+\tM_T(\vart)
	\end{eqnarray*}
	with
	\begin{eqnarray}\label{sampleM_T:def}
	M_T^0:=\sum_{|k| < \xi_T} \frac{1}{T}\sum_{t} \big|{Y}^{t}_{k}\big|^2\,,\quad \tM_T(\vart) &:=&  - \sum_{|k| < \xi_T}
	\left|\frac{1}{T} \sum_{t} h_k(\delta_t^\vart) {Y}^{t}_{k}  \right|^2,
	\end{eqnarray}
	where $M_T^0$ does not depend on $\vart$. Similarly, we have the  \emph{population contrast functional}
	\begin{eqnarray*}
	M(\vart)  &:=&  \sum_{k \in \Z^2} \int_0^1
	\left| h_k(\delta^\vart_t-\delta^{\vart_0}_t)\, {f}_{k}
	- \int_0^1 h_k(\delta^\vart_{t'}-\delta^{\vart_0}_{t'})\, {f}_{k}
	dt'\right|^2 dt \\
	&=& \sum_{k \in \Z^2} |{f}_{k}|^2 \left(1- \left| \int_0^1 h_k(\delta^\vart_t-\delta^{\vart_0}_t)\, dt \right|^2 \right)~=~M^0+\tM(\vart)
	\end{eqnarray*}
	with
	\begin{eqnarray}\label{populationM_T:def}
	M^0:=\sum_{k \in \Z^2}|{f}_{k}|^2
	\,,\quad
 	\tM(\vart) &:=&  - \sum_{k \in \Z^2} |f_k|^2 \left| \int_0^1 h_k(\delta^\vart_t-\delta^{\vart_0}_t)\, dt \right|^2
	\end{eqnarray}
	where $M^0$ is a constant in $\vart$. Note that while the population contrast involves the unknown true image $f$ and true parameter $\vart_0$, the empirical contrast only involves the data and the model. This gives rise to the following estimator.
	\begin{Def}\label{estimator:def}
	For given $T>0$ and choice of $\xi_T>0$ define an estimator for the parameter of the drift function by
	$$\hat\vart_T \in \arg \min_{\vart \in \Theta} M_T(\vart)$$
	and the corresponding estimator for the image $f$ as	
	$$\hat f_T(x) := \sum_{|k| < \xi_T} \frac{1}{T} \sum_{t\in\T} h_k(\delta_t^{\hat\vart_T})\,Y^{t}_k\;  e^{2\pi i \left\langle k , x\right\rangle}\,.$$
	\end{Def}
	Obviously, the proposed estimator is closely related to the concept of M-estimators. To derive the asymptotics below, we will equivalently maximize $-\tM_T(\vart)$ as well as $-\tM(\vart)$ since their difference is constant in $\vart$.

\subsection{Main Results}\label{main:sec}

	Recall the definition of a Sobolev space of order $\rho>0$, e.g. \cite[p.245]{Evans1998}.
	$$H^{\rho}([0,1]^2) := \left\{f\in L^1([0,1]^2): \sum_{k \in \Z^2} (1+ |k|^2)^{\rho} |{f}_k|^2 < \infty\right\}\,.$$

\teacher{
	From the Plancherel identity we infer at once that $H^{\rho'}([0,1]^2)\subset H^{\rho}([0,1]^2)\subset L^2([0,1]^2)\subset L^1([0,1]^2)$ for all $\rho'\geq\rho>0$.
}

	Additionally to Assumptions \ref{image-drift:as} on image and drift we require the following assumptions for consistency of the estimator $\hvart_T$.

	The following assumption that there be combinations of indices which are relatively coprime  with non-vanishing Fourier coefficients allows to deduce uniqueness of the minimizer in (\ref{populationM_T:def}).

	\begin{As}
	\label{f-Sobolev-half:as}
	Let $f \in H^{1}([0,1]^2)$ and suppose there exist some $k_1,k_2,k'_1,k'_2,k''_1,k''_2,k'''_1,k'''_2\in \mathbb Z$ such that $k_1k'_2-k_2k'_1$ as well as $k''_1k'''_2-k''_2k'''_1$ are non-zero, have no common divisors  and $|f_k|\neq 0$ for all $k \in \left\{(k_1,k_2),(k'_1,k'_2),(k''_1,k''_2),(k'''_1,k'''_2)\right\}$.
	\end{As}

	\begin{Rms}\label{f-Fourier-coeff:rm}\quad
	\begin{enumerate}
		\item We need the property $f \in H^1([0, 1]^2)$ for the asymptotic normality of the estimator $\hvart_T$. For the consistency of $\hvart_T$ it is sufficient to have $f \in H^{1/2}([0, 1]^2)$, if we additionally assume $\sup_{k\in\mathbb Z^2}|{f}_k|\,|k|<\infty$.
		\item Every $f\in H^1([0,1]^2)$ satisfies $\sup_{k\in\mathbb Z^2}|{f}_k|\,|k|<\infty$ since
		\[ \infty > \sum_{k \in \Z^2} |{f}_k|^2\,(1+ |k|^2) \geq \sup_{k\in\mathbb Z^2}|{f}_k|^2\,|k|^2 = \left(\sup_{k\in\mathbb Z^2}|{f}_k|\,|k|\right)^2\,. \]
	\end{enumerate}
	\end{Rms}

	\begin{As}
	\label{delta-TV:as}
	The map
	$$\Theta \to L^1([0,1],[0,1]^2), \vart \mapsto \delta_t^\vart=\big((\delta_t^{\vart})_1,(\delta_t^{\vart})_2\big)$$
	is injective and continuous {w.r.t.} the norm $|\delta_t^\vart| = \big|(\delta_t^{\vart})_1\big| + \big|(\delta_t^{\vart})_2\big|$. Moreover for each $\vart$ the drift $\delta_t^\vart$ as a function in $t$ is continuous at $t=0$ and of bounded total variation in both components with bound $TV\big((\delta_t^{\vart})_1\big) + TV\big((\delta_t^{\vart})_2\big) < C$ for some constant $C>0$ uniformly in $\vart$, where $TV(h)$ denotes the total variation norm of a real valued function $h:[0,1]\to\R$.
	\end{As}

	\begin{As}
	\label{sigma:as}
	There is a constant $\sigma_{\textup{max}} \in (0, \infty)$ s.t.
	$\sigma_{j,t} \leq \sigma_{\textup{max}}$ for all $t \in \T$, $j \in J_t$.
	\end{As}

	\begin{As}
	\label{delta-Lipschitz:as}
	The drift function $\delta_t$ is locally a uniformly Lipschitz function, i.e.
	for $\vart$ in a neighborhood of $\vart_0$ there is a constant $L>0$, such that
	\[ \sup_{t\in [0,1]} |\delta_t^\vart - \delta_t^{\vart_0}| \leq L \|\vart - \vart_0\|. \]
	\end{As}
	
	\begin{Th}\label{Consistency:Thm}
	Suppose that Assumptions \ref{image-drift:as}, \ref{f-Sobolev-half:as}, \ref{delta-TV:as} and \ref{sigma:as} hold. If we choose $\xi_T$ such that $\xi_T \stackrel{T \to \infty}{\longrightarrow} \infty$ and $\xi_T = o(\sqrt{T})$ then the drift estimator $\hvart_T$ from Definition \ref{estimator:def} is consistent, i.e.
	\begin{equation}\label{consistency1:eq}
	  \hat\vart_T \stackrel{T \to \infty}{\longrightarrow} \vart_0 \; \text{a.s.}
	\end{equation}
	If additionally Assumption \ref{delta-Lipschitz:as} holds, then also
	\begin{eqnarray}\label{eq:thmb}
	\left\|\hat f_T - f \right\|_2 \stackrel{T \to \infty}{\longrightarrow} 0  \mbox{ in probability, and if $\xi_T=o(T^{1/4})$ then }	\left\|\hat f_T - f \right\|_2 \stackrel{T \to \infty}{\longrightarrow} 0 \; \text{a.s.}  
	\end{eqnarray}
	\end{Th}

	The proof of this theorem is deferred to the Appendix.

\teacher{
	\noindent{\it Plan of Proof.} To prove consistency of $\hat\vart_T$ we follow a standard three step argument in M-estimation (e.g. \cite{Vaart2000} and \cite{GamboaLoubesMaza2007}). First we show the uniqueness of the population contrast minimizer $\vart_0$. In a second step we establish the continuity of $\vart\to\tM(\vart)$. Thirdly, we verify that $\tM_T(\vart) \to \tM(\vart)$ a.s. uniformly over $\vart\in\Theta$ as $T,\xi_T \to \infty$, $\xi_T=o(\sqrt{T})$. In consequence,  \cite[Theorem 5.7]{Vaart2000} (yielding weak consistency) can be adapted to obtain strong consistency. For convenience, here is the corresponding argument:

	\begin{quote}
	 Let $\varep >0$. If $\vart_0$ is a unique minimizer, by continuity of $\tM$ there is $\eta_\varep>0$ such that  $\tM(\vart)>\tM(\vart_0)+\eta_\varep$ for all $\vart\in\Theta$ with $\|\vart-\vart_0\|\geq \varep$. Hence
	 \begin{eqnarray*}
	P\Big(\mathop{\limsup}_{T\to\infty}\big\{\|\hat\vart_T -\vart_0\|\geq \varep\big\}\Big)
	&\leq & P\Big(\mathop{\limsup}_{T\to\infty}\big\{\tM(\hat\vart_T)>\tM(\vart_0)+\eta_\varep\big\}\Big)	 \\&=&P\Big\{\mathop{\limsup}_{T\to\infty}\tM(\hat\vart_T)>\tM(\vart_0)+\eta_\varep\Big\}~=~0\,,
      \end{eqnarray*}
      because $\tM_T(\hat\vart_T)- \tM(\vart_0)<0$ a.s. for $T$ sufficiently large and hence a.s.
	 \begin{eqnarray*}

	  \mathop{\limsup}_{T\to\infty}\tM(\hat\vart_T) - \tM(\vart_0) &\leq & \mathop{\limsup}_{T\to\infty}\Big(\tM(\hat\vart_T) - \tM_T(\hat\vart_T)\Big)
	  \\
	  &\leq&\mathop{\limsup}_{T\to\infty} \sup_{\vart \in\Theta} \Big(\tM(\vart) - \tM_T(\vart)\Big)~=~0\,.
	 \end{eqnarray*}
	\end{quote}

	In a fourth step, (\ref{eq:thmb}) is then established. The detailed proof comprising all four steps is provided for in the Appendix.\qed

}

\begin{Rm}
  \label{Rm:ChoiceOfCutoff}(Choice of $\xi_T$).  
  The finite sample behaviour of our estimator depends on the choice of $\xi_T$, which should be large enough to capture all important features of the image $f$, but not too large in order to filter out the noise. 
  Theorem \ref{Consistency:Thm} suggests $o(\sqrt{\xi_T})$.  We found numerically that thresholds in a relatively large range  performed equally well. Therefore, we used for simplicity  $\xi_T = \sqrt{T}$ in all our simulations (section \ref{sec:simulations}) and obtained always good results. We mention that comparable results were obtained for $\xi_T$'s,  $\xi_T = c \sqrt{T}$ where $c\in [0.2,1]$ (simulations not displayed) rendering the estimation process as quite robust w.r.t. to this parameter as long as it is not chosen too small.  In the ananlysis of SMS data (section \ref{sec:application}), we have chosen   $\xi_T = 100 = 1/2 \sqrt{T'}$.
\end{Rm}

In the following we show under twice differentiability of the drift asymptotic normality of our estimator.
	\begin{As}
	\label{f-Sobolov-one:as}
	Let $f \in H^1([0, 1]^2)$ and assume that there exists a neighborhood $U\subset \Theta\subset\mathbb R^d$ of $\vart_0$ and some $R>0$ such that $\vart\mapsto \delta_t^\vart$ is  twice differentiable in $U$ for all $t\in [0,1]$ such that for $r=1,2$
	$$\big\|\grad_\vart\big((\delta_t^{\vart})_r\big)\big\|, \big\|\Hess_\vart\big((\delta_t^{\vart})_r\big)\big\|<R$$
	and the second partial derivatives 
	are continuous at $\vart_0$. Also assume that both components of every partial derivative $t\mapsto \partial_{\vart_j}\delta_t^{\vart}$ ($1\leq j\leq d$), $\vart=(\vart_1,\dotsc,\vart_d)$ are of bounded total variation on $[0,1]$ at $\vart=\vart_0$ (cf. Assumption \ref{delta-TV:as}).
	\end{As}

	\begin{As}
	\label{convergent-sigma:as}
	For all $k \in \Z^2$ define
	\begin{align*}
	  (\bar\tau_k^T)^2 &:= \frac1T \sum_{t \in \T} \frac1{n_t} \sum_{j \in J_t} \sigma_{j, t}^2 \cos(-2\pi\langle k, x_{j,t}-\delta_t^{\vart_0}\rangle)^2,\\
	  (\bar\omega_k^T)^2 &:= \frac1T \sum_{t \in \T} \frac1{n_t} \sum_{j \in J_t} \sigma_{j, t}^2 \sin(-2\pi\langle k, x_{j,t}-\delta_t^{\vart_0}\rangle)^2.
	\end{align*}
	We have $(\bar\tau_k^T)^2, (\bar\omega_k^T)^2 > 0$ and there are $\sigma_{A, k}^2, \sigma_{B, k}^2 > 0$ such that
	\begin{equation}
	  \label{convergent-sigma:eq}
	  (\bar\tau_k^T)^2 \to \sigma_{A, k}^2, \quad (\bar\omega_k^T)^2 \to \sigma_{B, k}^2 \quad \text{uniformly in } k \text{ as } T \to \infty.
	\end{equation}
	\end{As}

	If Assumption \ref{f-Sobolov-one:as} is satisfied the following matrices are well defined. Here $\grad_\vart\langle k,\delta_t^{\vart_0}\rangle$ denotes the gradient evaluated at $\vart_0$ and $\grad'_\vart\langle k,\delta_t^{\vart_0}\rangle$ denotes its transpose.
    \[ \Sigma := \sum_{k\in \Z^2}|f_k|^2 q_k, \quad
       \tilde\Sigma := \sum_{k\in \Z^2}\bigl(\sigma_{A, k}^2\RE(f_k)^2+\sigma_{B, k}^2\IM(f_k)^2\bigr) q_k, \]
    where
    \[ q_k := \int_0^1\grad_\vart\langle k,\delta_t^{\vart_0}\rangle\,\grad'_\vart\langle k, \delta_t^{\vart_0}\rangle\,dt
	-
	\mathop{\dint}_{[0,1]^2}\grad_\vart\langle k,\delta_t^{\vart_0}\rangle\,\grad'_\vart\langle k, \delta_{t'}^{\vart_0}\rangle\,dt\,dt'. \]
	Note, that $\tilde\Sigma$ is positive definite iff $\Sigma$ is, because $\bigl(\sigma_{A, k}^2\RE(f_k)^2+\sigma_{B, k}^2\IM(f_k)^2\bigr) \neq 0$ iff $|f_k|^2 \neq 0$.
	
	\begin{Th}\label{asymptotic-normality:thm} Under Assumptions \ref{f-Sobolov-one:as} and \ref{convergent-sigma:as} with the notation from (\ref{sampleM_T:def}) as $T,\xi_T\to \infty$ with $\xi_T = o(T^{1/4})$, we have that
	\begin{enumerate}
	\item[(i)] $\sqrt{T}\,\grad_\vart M_T(\vart_0)$ tends asymptotically to a $d$-variate normally distributed random vector with zero mean and covariance matrix $ 16\pi^2 \tilde\Sigma$.
	\item[(ii)] ${\rm Hess}\tM_T(\vart_0) \to 8\pi^2\Sigma$ a.s.
	\end{enumerate}

	\end{Th}
	The proof of this theorem is deferred to the Appendix.

	\begin{Th}\label{clt:thm}
	 Under Assumptions \ref{f-Sobolev-half:as}, \ref{delta-TV:as}, \ref{f-Sobolov-one:as} and \ref{convergent-sigma:as} if $\hat{\vart}_T {\to}\vart_0$ a.s. and $\xi_T/T^{1/4}\to 0$, then
	\[\Sigma \sqrt{T}(\hat{\vart}_T-\vart_0) ~{\konvD}~ {\cal N}\left(0,\frac{1}{4\pi^2}\,\tilde\Sigma\right) \quad \text{as } T \to \infty\]
	in distribution. In particular, if $\Sigma$ is of full rank then
	\begin{equation}
        \label{CLT:eq}
        \sqrt{T}(\hat{\vart}_T-\vart_0) ~{\konvD}~ {\cal N}
	\left(0,\frac{1}{4\pi^2}\,\Sigma^{-1}\tilde\Sigma\Sigma^{-1}\right) \quad \text{as } T \to \infty.
	\end{equation}
	\end{Th}
	
	\begin{proof}
	Under Assumption \ref{f-Sobolov-one:as}, standard expansion arguments from M-estimation can be used as follows.
	Since $M_T(\vart)$ is twice continuously differentiable for $\vart$ near $\vart_0$ and $\hvart_T$ converges a.s. to $\vart_0$, we have that
	\begin{eqnarray*}
	0 &=& \grad_\vart M_T(\hvart_T) \\&=& \grad_\vart M_T(\vart_0) + \Hess_\vart M_T(\vart_0) (\hvart_T-\vart_0) + \Big(\Hess_\vart M_T(\hvart^*)-\Hess_\vart M_T(\vart_0)\Big)(\hvart_T-\vart_0)
	\end{eqnarray*}
	where $\hvart^*$ lies between $\vart_0$ and $\hvart_T$.
	The continuity of the second derivatives gives that $\hvart_T-\vart_0$ and $\grad_\vart M_T(\vart_0)$ are of the same asymptotic order since $\Hess_\vart M_T(\vart_0) \to 8\pi^2\Sigma$ a.s. holds by (ii) of Theorem \ref{asymptotic-normality:thm}.
	Hence
	$$8\pi^2\Sigma(\hvart_T-\vart_0) = - \,\grad_\vart M_T(\vart_0) + o_P(\| \hvart_T-\vart_0 \|)$$
	which in conjunction with (i) of Theorem \ref{asymptotic-normality:thm}, yields both asymptotic assertions.
	\end{proof}

	\begin{Ex}[Linear drift]
	For linear drift $\delta_t^{\vart} = \vart t$, we have $\grad_\vart \langle k, \delta_t^{\vart}\rangle = kt$. Thus,
	      \[ \Sigma = \frac1{12}\sum_{k \in \Z^2} |f_k|^2 \left(\begin{array}{cc}
k_1^2 & k_1 k_2 \\ 
k_1 k_2 & k_2^2
\end{array} \right), \quad \det(\Sigma) = \frac1{144} \sum_{k, l \in \Z^2} |f_k|^2 |f_l|^2 (k_1^2 l_2^2 - k_1 k_2 l_1 l_2). \]
          If $\det(\Sigma) \neq 0$, i.e. $\Sigma > 0$ (which is the case iff $f$ is not constant in any direction; see Lemma 7.2 in the Appendix), we can calculate $\Sigma^{-1}$ and get (\ref{CLT:eq}) with
          \begin{align*}
\Sigma^{-1} &= \frac{1}{12\det(\Sigma)}\sum_{k \in \Z^2} |f_k|^2 \left(\begin{array}{cc}
k_2^2 & -k_1 k_2 \\ 
-k_1 k_2 & k_1^2
\end{array} \right), \\
\tilde\Sigma &= \frac1{12}\sum_{k \in \Z^2} \bigl(\sigma_{A, k}^2\RE(f_k)^2+\sigma_{B, k}^2\IM(f_k)^2\bigr) \left(\begin{array}{cc}
k_1^2 & k_1 k_2 \\ 
k_1 k_2 & k_2^2
\end{array} \right).
          \end{align*}
        \end{Ex}

        \begin{Rms}\label{rem:2.15}\quad \begin{enumerate}
        \item\label{rem:2.15a}
	Although the estimate $\hvart_T$ does not rely on knowledge of the local variances $\sigma_{j, t}^2$, they occur in the limiting variance of Theorem \ref{asymptotic-normality:thm}.
	To estimate the variance $\sigma^2$ in case of constant $\sigma_{j,t}^2 \equiv \sigma^2$ in (\ref{few-observations:model}) one can use simple difference based estimates (see \cite{Munk2005} and the references given there).
	Note that in the Poisson model underlying (\ref{few-observations:model}) an approximately constant variance can be achieved by employing a variance stabilizing transformation, e.g. $\sqrt{Z_{j, t}+1/4}$ (see e.g. \cite[page 378]{Frick2013} for a careful description).
	In the case of nonconstant $\sigma_{j,t}^2$, sufficiently smooth and bounded away from zero, an estimator of these quantities can be obtained in general from a nonparametric variance estimator (see \cite{BrownLevine2007} and the references given therein). However, note that for the limiting variances $\sigma_{A, k}^2$, $\sigma_{B, k}^2$ simpler estimators can be employed, based on proper spatial differences along similar lines as in \cite{Dette1998}.
	\item In particular, $\hat\vart_T - \vart_0$ has the parametric rate $T^{-1/2}$ if $\xi_T$ is chosen to be fixed, although the nuisance parameter $f$ (see Assumption \ref{f-Sobolev-half:as}) is infinite dimensional. This can be interpreted in the sense that a finite number of Fourier coefficients are sufficient for detection of the drift parameter, which reduces the problem to finite dimensions, a well known phenomenon from semiparametric estimation \cite{Schick2008, Bickel1998, Vaart2000}.
	It is a challenging problem to derive the semiparametric efficient estimator for $\delta_t$ in model \ref{few-observations:model} and its asymptotic distribution. One reason for the Fourier based approach adopted here is that a time shift simply results in a multiplication in the Fourier domain (see (\ref{shift:prop})) which simplifies the statistical analysis. However, we expect that our estimator will not be semiparametrically efficient, although the $\sqrt{T}$ rate of convergence appears to be optimal.
	\end{enumerate}
	\end{Rms}

\section{Simulations}\label{sec:simulations}

To investigate the finite sample properties of the proposed method we conduct a simulation study\footnote{An R-package providing the software for the simulations as well as the application to SMS data is available at \texttt{www{.}stochastik{.}math{.}uni-goettingen{.}de/R\_ImageDrift}.} with images of size $n=N^2$ pixels with $N=256$. We opt for $T\in \{20,50,100\}$ in order to reduce computational time, as our implementation requires several minutes to compute $\hvart_T$ on a $256\times 256$ image for $T=1000$, say.
In order to make this simulations comparable to the data in section \ref{sec:application} we choose drift parameters $\theta_0$ such that the total drift (i.e. the pixel shift between the first an the last image) has comparable scale to the ones observed in our SMS data.

We consider the model (\ref{few-observations:model}) with four different drift types: linear, quadratic, and cubic drift, as well as a piecewise linear drift with a jump at unknown time. Note that the drift with jump violates the Lipschitz property in Assumption \ref{delta-Lipschitz:as}. To ensure that the multiplication in the Fourier domain corresponds to an integer valued shift of pixels in the image domain we consider rounded versions of the drift function $\tilde\delta_t(\vart)_1 = \lfloor N \delta_t(\vart)_1 + 0.5\rfloor/N$ and $\tilde\delta_t(\vart)_2 = \lfloor N \delta_t(\vart)_2 + 0.5\rfloor/N$.

We use the test image displayed in Figure \ref{fig:testimage}, with image intensity $f$ ranging from zero to one (the average image intensity is about 0.045), and apply three error models.

\begin{figure}[htbp]
  \centering
\ifthenelse{\equal{\user}{alex}}{
  \includegraphics[width=7.5cm]{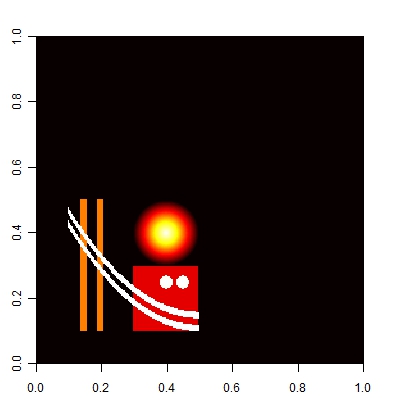}
  }{}
\ifthenelse{\equal{\user}{stephan}}{
  \includegraphics[width=7.5cm]{../IMG/trueimage.jpg}
  }{}
  \caption{\it Test image $f$ with grey scale values (rescaled to the unit interval), represented by colours ranging from black (0) over red and yellow to white (1).}
  \label{fig:testimage}
\end{figure}

We aim to apply our method to SMS microscopy and therefore, following the model (\ref{few-observations:model}), introduce a randomness of the selected pixel locations at each time point, such that every pixel of the original image contributes information exactly once. For every pixel location $x_j = \bigl((x_j)_1, (x_j)_2)$, $j \in \{1, \dotsc, N^2\}$, we randomly select a time point $t_j$ via the uniform distribution on $\T = \{0, 1/T, \dotsc, (T-1)/T\}$ (such that the $t_j$ are independent). Then we observe the (noisy) value $f(x_j)$ at time $t$ and location $x_j+\tilde\delta_t(\vart_0)$ if and only if $t=t_j$, otherwise we observe nothing (i.e. the value 0) at that time and pixel location. Note that $\T$ is defined in such a way that the whole time interval is always $[0, 1]$, i.e. if $T=100$, the time between two subsequent frames is exactly half as long as if $T$ were only 50.

First, a Gaussian error model (see (\ref{few-observations:model}))
\[ Z_{j, t} := Z\bigl((x_j)_1+\tilde\delta_t(\vart_0)_1, (x_j)_2+\tilde\delta_t(\vart_0)_2\bigr) 
  := \begin{cases} f(x_j) + \sigma\varep_{j,t} & \text{if } t=t_j,\\ 0 & \text{if } t\neq t_j \end{cases} \]
with $\sigma > 0$ and i.i.d. standard normal random variables $\varep_{j,t}$.

Secondly, in order to illustrate robustness of our estimation method against outliers, we assume that the $\varep_{j, t}$ are i.i.d. $t$-distributed with 2 degrees of freedom.

Finally, we simulate a Poisson error model, where the $Z_{j, t}$ are mutually independent and (given that $t=t_j$) Poisson distributed with intensity $f(x_j)$.
As mentioned in Remark \ref{rem:2.15}{.}\ref{rem:2.15a}, we use a variance stabilizing transformation $\sqrt{Z_{j, t}+1/4}$.
We minimize the discretized version of the contrast functional (\ref{eqn:contrast}) and use fast Fourier transform (FFT)
\teacher{
which here takes the form
$$
\tM_T(\vart) =
\sum_{k_1 =-\left\lfloor \xi_T /2\right\rfloor}^{\left\lfloor \xi_T /2\right\rfloor}
\sum_{k_2 =-\left\lfloor \xi_T /2\right\rfloor}^{\left\lfloor \xi_T /2\right\rfloor}
\left|\frac{1}{T} \sum_{t=0}^{(T-1)/T} \hat{Y}^{t}_{(k_1,k_2)} e^{2\pi i ( \delta_{t}(\vart)_1 k_1/N_1 + \delta_{t}(\vart)_2 k_2/N_2)} \right|^2,
$$
where
\[ \hat{Y}^{t}_{(k_1,k_2)} = \frac{1}{N_1 N_2} \sum_{x_1=1/N_1}^{1} \sum_{x_2=1/N_2}^{1} y_{(x_1, x_2), t}
 e^{-2\pi i ( x_1 k_1 + x_2 k_2)} \]
denotes the discrete Fourier transform of the noisy images,
}
which can be performed in $O(N^2 \cdot 2\log(N))$ steps.
The unique minimizer (cf. Step I of the proof of Theorem \ref{Consistency:Thm} which is detailed in the Appendix) is evaluated by a standard Nelder-Mead-type algorithm as implemented in the statistical software \textbf{R}.
We specify the parameters as follows: $\sigma = 0.1$ (Gaussian and $t_2$-distributed errors), $\xi_T = \sqrt{T}$.
As start value for the minimization algorithm we choose $0 \in \R^d$, where $d$ is the dimension of the drift parameter $\vart_0$.

\begin{sidewaystable}
\begin{center}
{\footnotesize
\begin{tabular}{lr|c|c|c|c}
  \hline
 & & linear drift & quadratic drift & cubic drift & drift with jump \\
  \hline
true parameter & $\vart_0$ & $( 0.195 , 0.117 )$ & $( 0.195 , 0.039 , 0 , 0.078 )$ & $( 0.195 , 0 , 0.039 , 0 , 0.039 , 0.195 )$ & $( 0.312 , 0.312 , 0.156 , 0.312 , 0.156 , 0.234 , 0.5 )$ \\
   \hline
\hline
error type & $T$ & $\hvart_T$ & $\hvart_T$ & $\hvart_T$ & $\hvart_T$ \\
   \hline
Gaussian & 20 & $( 0.191 , 0.115 )$ & $( 0.179 , 0.053 , 0.022 , 0.054 )$ & $( 0.135 , 0.108 , -0.021 , -0.018 , 0.091 , 0.153 )$ & $( 0.377 , 0.277 , 0.183 , 0.285 , 0.059 , 0.241 , 0.5 )$ \\ 
  & 50 & $( 0.195 , 0.121 )$ & $( 0.201 , 0.039 , 0.001 , 0.083 )$ & $( 0.191 , 0.015 , 0.027 , -0.006 , 0.056 , 0.184 )$ & $( 0.329 , 0.326 , 0.16 , 0.337 , 0.228 , 0.224 , 0.53 )$ \\ 
  & 100 & $( 0.2 , 0.119 )$ & $( 0.188 , 0.049 , 0.011 , 0.069 )$ & $( 0.225 , -0.064 , 0.085 , 0.003 , 0.022 , 0.204 )$ & $( 0.32 , 0.285 , 0.162 , 0.313 , 0.2 , 0.222 , 0.48 )$ \\ 
   \hline
$t$-distr. & 20 & $( 0.189 , 0.119 )$ & $( 0.169 , 0.062 , 0.053 , 0.016 )$ & $( 0.144 , 0.084 , 0.002 , -0.018 , 0.097 , 0.145 )$ & $( 0.302 , 0.376 , 0.147 , 0.286 , 0.023 , 0.266 , 0.54 )$ \\ 
  & 50 & $( 0.193 , 0.123 )$ & $( 0.186 , 0.046 , 0.029 , 0.056 )$ & $( 0.194 , 0.034 , 0.004 , 0.015 , 0.083 , 0.134 )$ & $( 0.321 , 0.304 , 0.159 , 0.313 , 0.152 , 0.241 , 0.51 )$ \\ 
  & 100 & $( 0.203 , 0.114 )$ & $( 0.168 , 0.072 , 0.052 , 0.022 )$ & $( 0.205 , -0.041 , 0.078 , 0.017 , 0.07 , 0.146 )$ & $( 0.342 , 0.384 , 0.141 , 0.3 , 0.152 , 0.234 , 0.5 )$ \\ 
   \hline
Poisson & 20 & $( 0.183 , 0.127 )$ & $( 0.197 , 0.03 , 0.016 , 0.075 )$ & $( 0.116 , 0.157 , -0.054 , 0.01 , 0.086 , 0.141 )$ & $( 0.268 , 0.352 , 0.148 , 0.337 , 0.094 , 0.263 , 0.54 )$ \\ 
  & 50 & $( 0.203 , 0.11 )$ & $( 0.172 , 0.062 , 0.002 , 0.076 )$ & $( 0.181 , 0.04 , 0.013 , -0.004 , 0.02 , 0.212 )$ & $( 0.361 , 0.293 , 0.182 , 0.318 , 0.109 , 0.245 , 0.53 )$ \\ 
  & 100 & $( 0.193 , 0.124 )$ & $( 0.151 , 0.081 , 0.031 , 0.047 )$ & $( 0.147 , 0.071 , 0.009 , -0.006 , 0.061 , 0.179 )$ & $( 0.285 , 0.317 , 0.155 , 0.325 , 0.192 , 0.226 , 0.5 )$ \\ 
   \hline
\end{tabular}
}
\caption{\it Displaying the estimated $\hvart_T$ for one simulation in different drift models. We have considered image sequences with $T\in \{20,50,100\}$ time points as well as Gaussian and Student-$t_2$ error models with variance $0.1^2$ and a Poisson model as explained in detail in the text.}
\label{lin.quadratic.cubic.reconstr.tab}
\end{center}

    \vspace*{0.2cm}

\begin{center}
{\footnotesize
\begin{tabular}{lr|c|c|c|c}
  \hline
 &  & linear drift & quadratic drift & cubic drift & drift with jump \\
  \hline
true parameter & $\vart_0$ & $( 0.195 , 0.117 )$ & $( 0.195 , 0.039 , 0 , 0.078 )$ & $( 0.195 , 0 , 0.039 , 0 , 0.039 , 0.195 )$ & $( 0.312 , 0.312 , 0.156 , 0.312 , 0.156 , 0.234 , 0.5 )$ \\
   \hline
\hline
error type & $T$ & mean of est's & mean of est's & mean of est's & mean of est's \\
   \hline
Gaussian & 20 & $( 0.196 , 0.116 )$ & $( 0.179 , 0.056 , 0.027 , 0.051 )$ & $( 0.151 , 0.081 , -0.001 , 0.003 , 0.064 , 0.166 )$ & $( 0.311 , 0.316 , 0.161 , 0.314 , 0.162 , 0.235 , 0.522 )$ \\ 
  & 50 & $( 0.195 , 0.117 )$ & $( 0.182 , 0.052 , 0.019 , 0.06 )$ & $( 0.177 , 0.037 , 0.017 , -0.004 , 0.074 , 0.162 )$ & $( 0.314 , 0.311 , 0.16 , 0.316 , 0.164 , 0.234 , 0.51 )$ \\ 
  & 100 & $( 0.195 , 0.117 )$ & $( 0.178 , 0.056 , 0.015 , 0.064 )$ & $( 0.168 , 0.037 , 0.026 , -0.001 , 0.07 , 0.164 )$ & $( 0.321 , 0.31 , 0.159 , 0.305 , 0.16 , 0.231 , 0.5 )$ \\ 
   \hline
$t$-distr. & 20 & $( 0.195 , 0.114 )$ & $( 0.177 , 0.056 , 0.028 , 0.05 )$ & $( 0.154 , 0.071 , 0.005 , -0.011 , 0.085 , 0.155 )$ & $( 0.305 , 0.311 , 0.171 , 0.303 , 0.166 , 0.236 , 0.517 )$ \\ 
  & 50 & $( 0.195 , 0.117 )$ & $( 0.182 , 0.052 , 0.022 , 0.056 )$ & $( 0.177 , 0.034 , 0.021 , -0.001 , 0.07 , 0.163 )$ & $( 0.312 , 0.312 , 0.16 , 0.313 , 0.161 , 0.234 , 0.509 )$ \\ 
  & 100 & $( 0.196 , 0.116 )$ & $( 0.176 , 0.058 , 0.016 , 0.063 )$ & $( 0.167 , 0.046 , 0.018 , -0.001 , 0.068 , 0.166 )$ & $( 0.311 , 0.309 , 0.159 , 0.316 , 0.157 , 0.235 , 0.506 )$ \\ 
   \hline
Poisson & 20 & $( 0.196 , 0.116 )$ & $( 0.174 , 0.06 , 0.021 , 0.057 )$ & $( 0.157 , 0.063 , 0.012 , 0.001 , 0.075 , 0.156 )$ & $( 0.311 , 0.317 , 0.162 , 0.314 , 0.155 , 0.237 , 0.524 )$ \\ 
  & 50 & $( 0.195 , 0.117 )$ & $( 0.174 , 0.06 , 0.021 , 0.057 )$ & $( 0.171 , 0.045 , 0.017 , 0 , 0.077 , 0.154 )$ & $( 0.322 , 0.31 , 0.164 , 0.313 , 0.156 , 0.235 , 0.514 )$ \\ 
  & 100 & $( 0.196 , 0.117 )$ & $( 0.176 , 0.058 , 0.024 , 0.055 )$ & $( 0.172 , 0.033 , 0.028 , -0.006 , 0.082 , 0.155 )$ & $( 0.312 , 0.317 , 0.157 , 0.314 , 0.159 , 0.233 , 0.506 )$ \\ 
   \hline
\end{tabular}
}
\caption{\it Setting as in Table \ref{lin.quadratic.cubic.reconstr.tab}. Displaying the means of the estimators $\hvart_T$ from 100 simulations each.}
\label{lin.quadratic.cubic.mean.tab}
\end{center}
\end{sidewaystable}

\begin{sidewaystable}
\begin{center}
\begin{tabular}{r|ccc|ccc|ccc}
  \hline
 &\multicolumn{3}{c|}{\mbox{Gaussian noise}}&\multicolumn{3}{c|}{\mbox{$t_2$ noise}}&\multicolumn{3}{c}{\mbox{Poisson model}}\\
  \hline
 & $T= 20 $ & $T= 50 $ & $T= 100 $ & $T= 20 $ & $T= 50 $ & $T= 100 $ & $T= 20 $ & $T= 50 $ & $T= 100 $ \\
  \hline
Linear drift & 6 & 5 & 5 & 26 & 6 & 8 & 9 & 8 & 7 \\ 
  Quadratic drift & 63 & 48 & 44 & 66 & 54 & 55 & 65 & 59 & 61 \\ 
  Cubic drift & 138 & 121 & 133 & 172 & 130 & 175 & 142 & 141 & 144 \\ 
  Drift with jump & 79 & 71 & 67 & 174 & 80 & 83 & 87 & 90 & 86 \\ 
    \hline
\end{tabular}
\caption{\it Thousandfold of the root of the mean squared error $\E ||\hvart_T - \vart_0||^2$ of the estimators $\hvart_T$ from 100 simulations each.}
\label{rmse:tab}

    \vspace*{1.5cm}

\begin{tabular}{lr|ccc|ccc|ccc}
  \hline
 &&\multicolumn{3}{c|}{\mbox{Gaussian noise}}&\multicolumn{3}{c|}{\mbox{$t_2$ noise}}&\multicolumn{3}{c}{\mbox{Poisson model}}\\
  \hline
 && $T= 20 $ & $T= 50 $ & $T= 100 $ & $T= 20 $ & $T= 50 $ & $T= 100 $ & $T= 20 $ & $T= 50 $ & $T= 100 $ \\
  \hline
  SI & Linear drift & $0.067$ & $0.050$ & $0.006$ & $-0.009$ & $-0.009$ & $-0.013$ & $0.011$ & $0.012$ & $-0.053$ \\ 
  & Quadratic drift & $-0.005$ & $0.011$ & $-0.019$ & $0.032$ & $-0.009$ & $-0.039$ & $-0.031$ & $-0.003$ & $-0.108$ \\ 
  & Cubic drift & $-0.024$ & $0.015$ & $-0.073$ & $-0.001$ & $0.002$ & $0.008$ & $-0.016$ & $0.048$ & $-0.041$ \\ 
  & Drift with jump & $0.013$ & $-0.034$ & $0.029$ & $0.007$ & $-0.015$ & $-0.015$ & $0.016$ & $0.031$ & $-0.055$ \\ 
  \hline
  $\hat{f}_T$ & Linear drift & $-0.679$ & $-0.842$ & $-0.707$ & $-0.205$ & $-0.102$ & $-0.192$ & $-0.387$ & $-0.318$ & $-0.338$ \\ 
  & Quadratic drift & $-0.411$ & $-0.447$ & $-0.432$ & $-0.147$ & $-0.060$ & $-0.128$ & $-0.205$ & $-0.188$ & $-0.179$ \\ 
  & Cubic drift & $-0.686$ & $-1.045$ & $-0.710$ & $-0.215$ & $-0.112$ & $-0.218$ & $-0.375$ & $-0.358$ & $-0.514$ \\ 
  & Drift with jump & $-0.201$ & $-0.326$ & $-0.582$ & $-0.096$ & $-0.217$ & $-0.072$ & $-0.078$ & $-0.123$ & $-0.289$ \\ 
   \hline
\end{tabular}
\caption{\it Blur measure values of the superimposed images (SI) and the estimated images $\hat{f}_T$. The corresponding estimators $\hvart_T$ are reported in Table \ref{lin.quadratic.cubic.reconstr.tab}. The images for cubic drift, drift with jump and $T \in \{20, 50\}$ are shown in Figures \ref{fig:sim.cubic} and \ref{fig:sim.jump}.}
\label{sim.blur.tab}
\end{center}
\end{sidewaystable}

\begin{figure}[htbp]
  \centering
\ifthenelse{\equal{\user}{alex}}{
  \includegraphics[width=14cm]{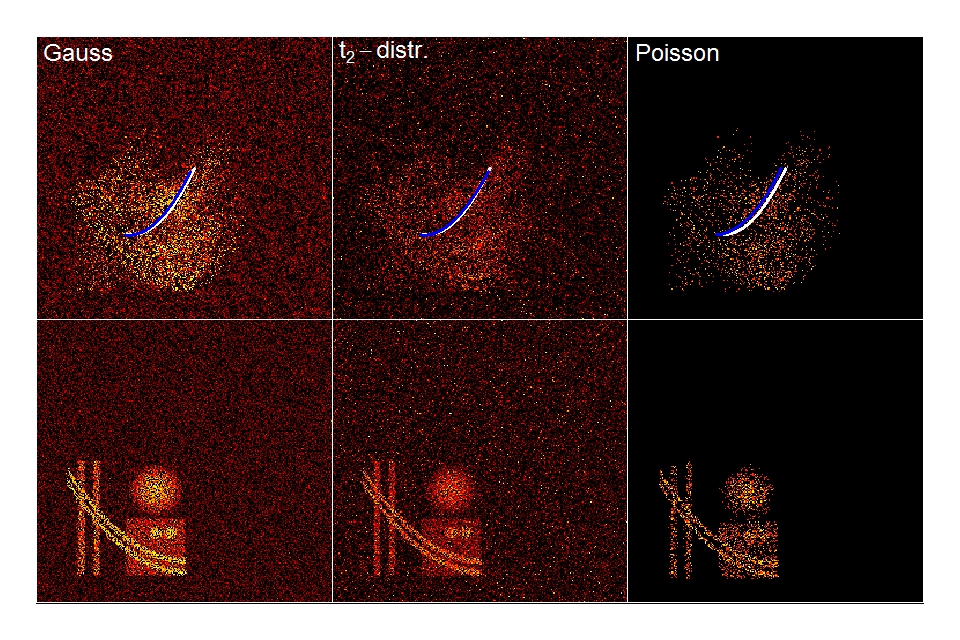}
  \includegraphics[width=14cm]{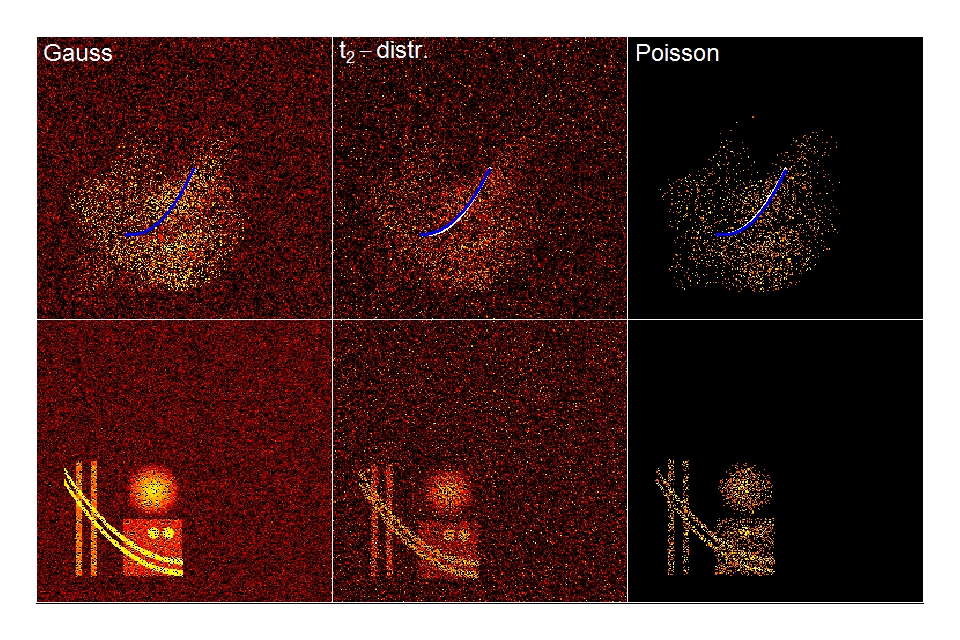}
  }{}
\ifthenelse{\equal{\user}{stephan}}{
  \includegraphics[width=14cm]{../IMG/gauss_t2_poisson_cubic_20.jpg}
  \includegraphics[width=14cm]{../IMG/gauss_t2_poisson_cubic_50.jpg}
  }{}
  \caption{\it The first row shows the superimposed images of sequences of $T=20$ noisy images subject to cubic drift (from left to right: Gaussian noise, Student-$t_2$ noise, Poisson model). The true drift curve of a single pixel is shown as a white curve segment on top of which we plot the estimated drift in blue. The true and the estimated parameters are reported in Table \ref{lin.quadratic.cubic.reconstr.tab}.
  Third row: The same with $T=50$ noisy shifted images. The second and fourth rows show the correspondingly reconstructed images.}
  \label{fig:sim.cubic}
\end{figure}

\begin{figure}[htbp]
  \centering
\ifthenelse{\equal{\user}{alex}}{
  \includegraphics[width=14cm]{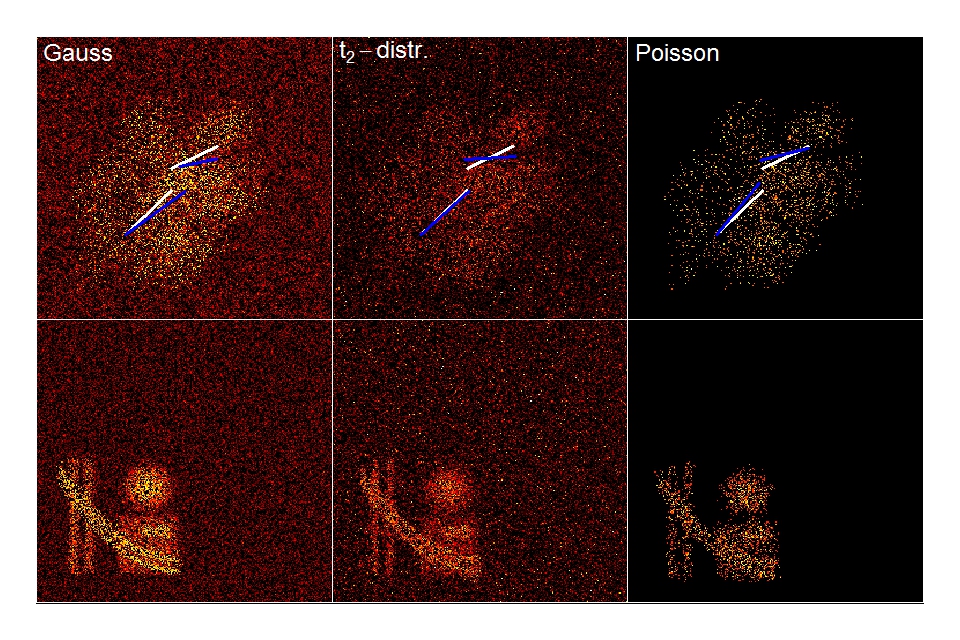}
  \includegraphics[width=14cm]{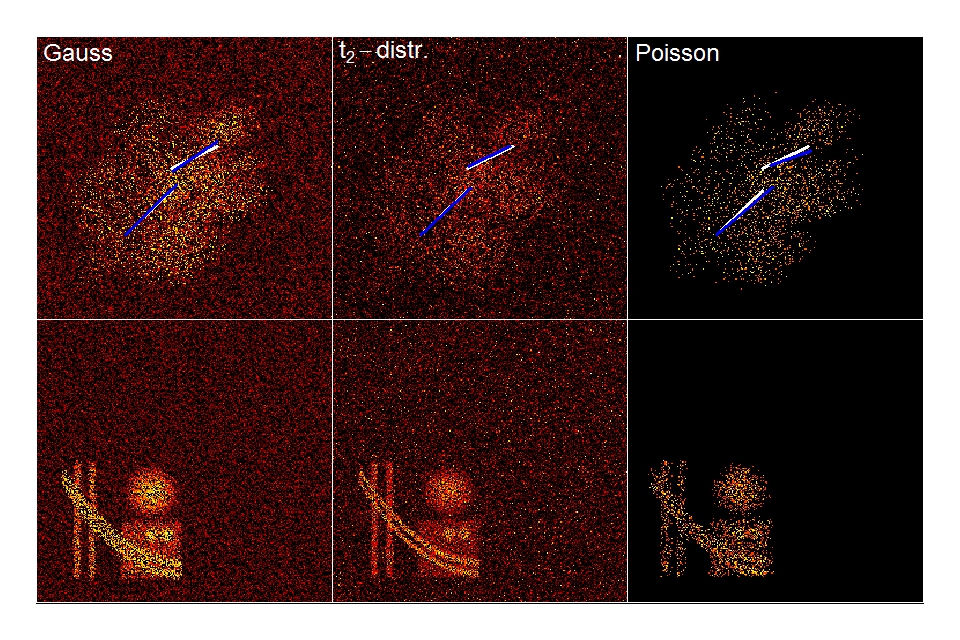}
  }{}
\ifthenelse{\equal{\user}{stephan}}{
  \includegraphics[width=14cm]{../IMG/gauss_t2_poisson_jump_20.jpg}
  \includegraphics[width=14cm]{../IMG/gauss_t2_poisson_jump_50.jpg}
  }{}
  \caption{\it The first row shows the superimposed images of sequences of $T=20$ noisy images subject to a piecewise linear drift with jump (from left to right: Gaussian noise, Student-$t_2$ noise, Poisson model). The true drift curve of a single pixel is shown as a white curve segment on top of which we plot the estimated drift in blue. The true and the estimated parameters are reported in Table \ref{lin.quadratic.cubic.reconstr.tab}.
  Third row: The same with $T=50$ noisy shifted images. The second and fourth rows show the correspondingly reconstructed images.}
  \label{fig:sim.jump}
\end{figure}

\paragraph{Polynomial drift models}
 have been described in Example  \ref{polynomial-drift:ex}. In the linear drift model we have $\delta_t(\vart) = \vart t$.
\teacher{
with $\vart=(\vart_1,\vart_2) \in \Theta \subset \R^2$ for a large closed square $\Theta$.
}
For the $x_1$-direction we choose $\vart_1= 50/256$, in $x_2$-direction $\vart_2=30/256$, i.e. the image is shifted by 50 pixels in $x_1$-direction and by 30 pixels in $x_2$-direction over the time interval $[0, 1]$ which, for $T=20$, translates to velocities of 2.5 and 1.5 pixels per frame, respectively, and so on.

In the quadratic drift model we set $\delta_t(\vart) = (\vart_{11}, \vart_{21})' t + (\vart_{12}, \vart_{22})' t^2$.
\teacher{
with $\vart=(\vart_{11},\vart_{12}, \vart_{21}, \vart_{22}) \in \Theta\subset\R^4$ with a large closed cube $\Theta$.
}
For the $x_1$-direction we choose $(\vart_{11}, \vart_{12}) = (50/256, 10/256)$, in $x_2$-direction $(\vart_{21}, \vart_{22}) = (0, 20/256)$.

Similarly we employ the cubic drift model $\delta_t(\vart) = (\vart_{11}, \vart_{21}) t + (\vart_{12}, \vart_{22}) t^2 + (\vart_{13}, \vart_{23}) t^3$.
\teacher{
with $\vart=(\vart_{11},\vart_{12}, \vart_{13}, \vart_{21}, \vart_{22}, \vart_{23}) \in  \Theta \subset \R^6$ with a large closed cube $\Theta$.
}
For the $x_1$-direction we choose $(\vart_{11}, \vart_{12}, \vart_{13}) = (50/256, 0, 10/256)$, in $x_2$-direction $(\vart_{21}, \vart_{22}, \vart_{23}) = (0, 10/256, 50/256)$.

The results of one estimate are reported in Table \ref{lin.quadratic.cubic.reconstr.tab}, the averages of 100 simulations in Table \ref{lin.quadratic.cubic.mean.tab}. As recorded in Table \ref{rmse:tab}, with increasing degree of the polynomials, the mean squared error increases. Still for the cubic drift model, visual inspection of the estimated images in Figure \ref{fig:sim.cubic} exhibits good reconstruction quality.

To evaluate our drift correction we use a version of the motion blur measure $m_2$ proposed in \cite{Xu2013} which is based on the work of \cite{Chen2010}. It is defined as
\begin{equation}
  \label{blur.measure:eq}
  m_2 := \log\left(\frac{J(\varphi_{\textup{max}})}{J(\varphi_{\textup{min}})}\right).
\end{equation}
Here, $J(\varphi) := \sum_{j=1}^{N^2} \Bigl(\Delta I\bigl((x_j)_1, (x_j)_2\bigr)_\varphi\Bigr)^2$ is the average squared directional derivative of an image $I$ in direction $\bigl(\cos(\varphi), \sin(\varphi)\bigr)'$, $\varphi \in [0, 2\pi)$, $\varphi_{\textup{min}}$ is the motion direction, and $\varphi_{\textup{max}}$ is the direction perpendicular to $\varphi_{\textup{min}}$. Note, that $J(\varphi) = 0$ iff $I$ is constant in direction $\varphi$. An advantage of $m_2$ is that it does not depend on the scale of the image. In \cite{Chen2010}, $\varphi_{\textup{min}}$ is selected as a minimizer of the functional $J$. The idea is that the image is blurred in the direction of the motion and thus the image intensity $f$ changes little in this direction (on average), while it varies much more in the perpendicular direction. The minimizer is obtained as follows:

Rewrite $J(\varphi) = \bigl(\cos(\varphi), \sin(\varphi)\bigr) D \bigl(\cos(\varphi), \sin(\varphi)\bigr)'$, where
\[ D = \left(\begin{array}{cc}
d_{11} & d_{12} \\
d_{12} & d_{22}
\end{array} \right), \quad d_{rs} := \sum_{j=1}^{N^2} \frac{\partial I}{\partial (x)_r}\bigl((x_j)_1, (x_j)_2\bigr)\cdot\frac{\partial I}{\partial (x)_s}\bigl((x_s)_j, (x_j)_2\bigr). \]
Then, $J(\varphi) = d_{11} + d_{12}\sin(2\varphi) + (d_{22} - d_{11})\bigl(\sin(\varphi)\bigr)^2$. We get the minimum value of $J$ by setting $dJ(\varphi)/d\varphi = d_12\cos(2\varphi) + (d_{22}-d_{11})\sin(2\varphi) = 0$, which yields $\varphi = \varphi_m + (r\pi)/2$, $r \in \Z$, with $\varphi_m = \arctan\bigl(2d_{12}/(d_{11}-d_{22})\bigr)/2$. The motion direction is then determined by
\[ \varphi_{\textup{min}} := \begin{cases}
                         \varphi_m & \text{if } J(\varphi_m) \leq J(\varphi_m + \pi/2),\\
                         \varphi_m + \pi/2 & \text{if } J(\varphi_m) > J(\varphi_m + \pi/2).
                       \end{cases} \]

The $J(\varphi_{\textup{max}})$ also keeps the blur measure value low in the case of an image that is (almost) constant over wide areas (where the directional derivative is small in any direction). Since we already know the true drift $\delta_t(\vart)$, we choose the average drift direction $\int_0^1 \partial \delta_t(\vart)/\partial t \, d t = \delta_1(\vart)$ as the motion direction (after normalization). Hence, in our context (where $I$ is either $\hat{f}_T$ or the superimposed image, see Table \ref{sim.blur.tab}) we get the motion blur measure
\begin{equation}
  \label{blur.measure:eq2}
  \tilde m_2 := \log\left(\frac{\sum_{j=1}^{N^2} \bigl\langle\grad_x I\bigl((x_j)_1, (x_j)_2\bigr), \textup{Rot}_{\pi/2}\delta_1(\vart)/|| \delta_1(\vart)||_2\bigr\rangle^2}{\sum_{j=1}^{N^2} \bigl\langle\grad_x I\bigl((x_j)_1, (x_j)_2\bigr), \delta_1(\vart)/||\delta_1(\vart)||_2\bigr\rangle^2}\right),
\end{equation}
where $||\cdot ||_2$ is the Euclidean norm and
\[ \textup{Rot}_{\pi/2} := \left(\begin{array}{rr}
\cos(\pi/2) & -\sin(\pi/2) \\
\sin(\pi/2) & \cos(\pi/2)
\end{array} \right) \]
is the rotation through $\pi/2$. We calculated an approximation of $\grad_x I$ as follows (see e.g. \cite{Gonzalez2002}).

Let $I$ be a pixel image of size $M \times N$. For every pixel location $(i, j)$, $i \in \{1, \dotsc, M\}$, $j \in \{1, \dotsc, N\}$, the gradient of $I$ is defined as $\nabla I(i, j) := \bigl(G_x(i, j), G_y(i, j)\bigr)'$ with
\[ G_x(i, j) := \sum_{i', j' = -1}^1 S_x(i'+2, j'+2) I(i+i', j+j'), \quad
   G_y(i, j) := \sum_{i', j' = -1}^1 S_y(i'+2, j'+2) I(i+i', j+j'), \]
where we extend the image periodically, i.e. $I(0, j) := I(M, j)$, $I(M+1, j) := I(1, j)$, $I(i, 0) := I(i, N)$, and $I(i, N+1) := I(i, 1)$ and so on. Here, $S_x$ and $S_y$ are the Sobel masks
\[ S_x := \frac18\left(\begin{array}{rrr}
-1 & 0 & 1 \\
-2 & 0 & 2 \\
-1 & 0 & 1
\end{array} \right), \quad S_y := \frac18\left(\begin{array}{rrr}
-1 & -2 & -1 \\
0 & 0 & 0 \\
1 & 2 & 1
\end{array} \right). \]
Often, especially if $I$ is noisy, it is beneficial to smooth the image first, e.g. with a Gauss kernel
\[ K := \frac1{16}\left(\begin{array}{rrr}
1 & 2 & 1 \\
2 & 4 & 2 \\
1 & 2 & 1
\end{array} \right). \]
This means that we replace every $I(i, j)$ with the weighted average
\[ \bar I(i, j) := \sum_{i', j' = -1}^1 K(i'+2, j'+2) I(i+i', j+j') \]
of the $3\times 3$ pixel area centred on it. Because our images are very noisy, we repeat that procedure once more.

\teacher{
This combination of smoothing and discrete differentiation can be done simultaneously by applying the Sobel type masks of size $7\times 7$,
\[ \bar S_x := \frac1{1280}\left(\begin{array}{rrrrrrr}
-1 & -4 & -5 & 0 & 5 & 4 & 1 \\
-6 & -24 & -30 & 0 & 30 & 24 & 6 \\
-15 & -60 & -75 & 0 & 75 & 60 & 15 \\
-20 & -80 & -100 & 0 & 100 & 80 & 20 \\
-15 & -60 & -75 & 0 & 75 & 60 & 15 \\
-6 & -24 & -30 & 0 & 30 & 24 & 6 \\
-1 & -4 & -5 & 0 & 5 & 4 & 1
\end{array} \right), \quad \bar S_y := \bar S'_x. \]
}

The motion blur values of the superimposed images and the corresponding estimated images are reported in Table \ref{sim.blur.tab}. The estimated image (i.e. with drift correction) is always less blurry than the superimposed image.

\teacher{
{\bf The following is not a good blur measure for our problem. Use the motion blur measure above!}
To evaluate our drift correction we use a standard blur measure from computer vision.
The following is a simplified version of the edge blur measure proposed in \cite{Gonzalez2002}.
Let $I$ be a pixel image of size $M \times N$. For every pixel location $(i, j)$, $i \in \{2, \dotsc, M-1\}$, $j \in \{2, \dotsc, N-1\}$, we calculate the corresponding gradient magnitude $G(i, j) := \sqrt{G_x(i, j)^2+G_y(i, j)^2}$, where
\[ G_x(i, j) := \sum_{i', j' = -1}^1 S_x(i'+2, j'+2) I(i+i', j+j'), \quad
   G_y(i, j) := \sum_{i', j' = -1}^1 S_y(i'+2, j'+2) I(i+i', j+j'), \]
with the Sobel masks
\[ S_x := \left(\begin{array}{ccc}
-1 & 0 & 1 \\
-2 & 0 & 2 \\
-1 & 0 & 1
\end{array} \right), \quad S_y := \left(\begin{array}{ccc}
-1 & -2 & -1 \\
0 & 0 & 0 \\
1 & 2 & 1
\end{array} \right). \]
If the image $I$ is blurry at $(i, j)$, i.e. if there is no sharp edge, then the gradient magnitude at $(i, j)$ is small. Otherwise, if the image has a sharp edge at $(i, j)$, the gradient magnitude is big. Thus, the multiplicative inverse of the average over all gradient magnitudes, the \emph{edge blur}
\begin{equation}
\label{blur.measure.old:eq}
\mu_e(I) := \left(\frac1{(M-2)(N-2)} \sum_{i=2}^{M-1} \sum_{j=2}^{N-1} G(i, j)\right)^{-1},
\end{equation}
is big for a blurry image and small for a clear image. The blur measure values of the superimposed images and the corresponding estimated images are reported in Table \ref{sim.blur.tab}. In the case of additive Gaussian or $t_2$ noise, the estimated image (i.e. with drift correction) is usually less blurry than the superimposed image. In contrast, in the case of the Poisson model, the blur measure value of the estimated image is always much bigger than the corresponding value of the superimposed image. This is due to the fact that in our Poisson model there is no noise present outside of the support of the image, i.e. even the noisy image is for the most part constant. This leads to a small average gradient altitude and in consequence to a large blur measure value. Since the area of constant image value of the corrected image is obviously much bigger than the one of the superimposed image, the blur measure falsely reports the first to be blurrier than the second.
}

\teacher{
\paragraph{Drift model with changepoint.}
We also consider a piecewise linear drift model with (unknown) changepoint
\[ \delta_t(\vart) = \begin{cases}
                        (\vart_{11}, \vart_{21})' t & \text{if } t \leq t_0 \\
                        (\vart_{12}, \vart_{22})' t + (\vart_{11} - \vart_{12}, \vart_{22} - \vart_{21})' t_0 & \text{if } t > t_0
                      \end{cases} \]
with $\vart_0=(\vart_{11}, \vart_{12}, \vart_{21}, \vart_{22}, t_0) \in \Theta\subset\R^5$, i.e. the drift changes direction at the unknown time point $t_0$.

For the simulation, we choose $(\vart_{11}, \vart_{21}) = (50/256, 10/256)$, $(\vart_{12}, \vart_{22}) = (30/256, 70/256)$, and $t_0=0.35$. We estimate $\vart_0$ and $t_0$ by the estimator with the smallest contrast value.
Once again, we use the Gaussian noise, the $t$-distributed noise with 2 degrees of freedom, and the Poisson model. The overlaid shifted images as well as their reconstructions are visualized in Figure \ref{fig:sim.change}. In the ultimate column of Table \ref{lin.quadratic.cubic.reconstr.tab} the estimation results are summarized.
}

\paragraph{Drift model with jump.}
Finally, in order to analyse the robustness of our method, e.g. when a smooth drift abruptly jumps due to an external shock, we consider a piecewise linear drift model with a jump at an unknown time.
\[ \delta_t(\vart) = \begin{cases}
                        (\vart_{11}, \vart_{21})' t & \text{if } t \leq t_0 \\
                        (\vart_{12}, \vart_{22})' (t - t_0) + (\vart_{13}, \vart_{23})' & \text{if } t > t_0
                      \end{cases} \]
with $\vart_0=(\vart_{11}, \vart_{12}, \vart_{13}, \vart_{21}, \vart_{22}, \vart_{23}, t_0) \in \Theta\subset\R^7$, i.e. the drift function jumps to the point $(\vart_{13}, \vart_{23})'$ at the unknown time point $t_0$.
As mentioned before, this type of drift does not meet our assumptions, e.g. the Lipschitz property in Assumption \ref{delta-Lipschitz:as} is not fulfilled as one can easily see by perturbing the parameter $t_0$.

For the simulation, we choose $(\vart_{11}, \vart_{21}) = (80/256, 80/256)$, $(\vart_{12}, \vart_{22}) = (80/256, 40/256)$, $(\vart_{13}, \vart_{23}) = (40/256, 60/256)$, and $t_0=0.5$.
We estimate $\vart_0$ and $t_0$ by the estimator with the smallest contrast value.
Once again, we use the Gaussian noise, the $t$-distributed noise with 2 degrees of freedom, and the Poisson model. The overlaid shifted images as well as their reconstructions are visualized in Figure \ref{fig:sim.jump}. In the ultimate column of Table \ref{lin.quadratic.cubic.reconstr.tab} the estimation results are summarized.

Note that the average drift direction used to determine the motion blur (\ref{blur.measure:eq2}) in the case of a drift function with jump is (before normalization) $t_0 \delta_{t_0}(\vart) + (1-t_0) \bigl(\delta_1(\vart) - \lim_{t\searrow t_0}\delta_t(\vart)\bigr)$ instead of just $\delta_1(\vart)$. The resulting blur values are reported in Table \ref{sim.blur.tab}.

\paragraph{Computational time.}
For polynomial drift, simulating a sample and computing the estimates required between 2 and 7 seconds on an Intel Core i7-4800MQ with 2.7 GHz. For the drift with jump, we considered jump times $\hat t_0$ on the grid $\{ 2/100, \dotsc, 98/100 \}$ and, given $\hat t_0$, minimized the contrast functional w.r.t. $\vart$ to find the estimator for $(\vart_0, t_0)$ with overall minimal contrast. This leads to higher computational times between about 21 seconds (Gaussian noise, $T=20$) and 3 minutes (Poisson model, $T=100$).

\quad

Our simulations show that the proposed estimation method works well and significantly reduces blurring. This has been demonstrated for a polynomial drift even if we observe just a small part of the shifted image at every time point. It also behaves robust to non-normality. We have obtained good results for reconstruction with $t_2$-distributed noise and in a Poisson model. Finally, we studied the case of a piecewise linear drift with a jump at an unknown time point, i.e. a discontinuous drift. Although Assumption \ref{delta-Lipschitz:as} is not satisfied in this case, we found that even in this setting our estimator performs quite well.

\section{SMS Nanoscopy Data}\label{sec:application}
We demonstrate how the estimation method proposed in Section \ref{sec:est} can be used to process SMS nanoscopy data. In particular, we address suitable choices for the drift model $\delta_t(\vart)$ as well as computational issues.

We used a standard SMS-setup for this study (modified from \cite{Geisler2012}) which was equipped with a home-built stable sample holder ensuring that the sample drift is well below the expected average localization accuracy. The excitation and switching light beams were provided by continuous wave lasers  running at 532 nm (HB-Laser, Germany) and 371 nm (Cube 375-16C, Coherent Inc, USA). The excitation power density in the sample plane was 5 $kW/cm^2$. If necessary the power density of the switching light was ramped up from  0 to a few 10 $W/cm^2$.  The objective was a NA 1.2 60x water immersion lens (UPLSAPO 60XW, Olympus Deutschland GmbH, Germany) and the camera was an EMCCD camera (Ixon X3 897, Andor Technology, Northern Ireland).
The microtubule network ($\beta$-tubulin) of fixed Vero-cells was immuno-labelled with the caged fluorescent dye Abberior Cage 552 (Abberior GmbH, Germany) according to standard protocols. The fiducial markers (FluoSpheres, Invitrogen, USA) were incorporated into the sample by spincoating a polyvinyl alcohol-fluosphere solution.
For image acquisition a series of  $T'=40,000$ frames was taken with a frame exposure time of 15 ms, resulting in a total image acquisition time of about 10 minutes. During this time an experimental drift was applied by moving the sample with respect to the objective lens with two linear positioners (SLC-17, SmarAct GmbH, Germany) in steps of 1 nm in a controlled manner.

The lateral positions of the fluorescent markers were then calculated from the single frames by a mask-fitting of the respective Airy spot \cite{Thompson2002}. These locations were tabulated together with the respective time of detection $t'\in\{1,\dotsc,T'\}$.

We analyse two data sets (networks I and II) from Abberior Cage 552-labeled $\beta$-tubulin networks in fixed vero-cells. The position histogram of the first data set is shown in Figure \ref{data.reconstr:fig}. It contains 1,077,909 positions recorded in 40,000 frames, which are distributed over an area of about $(9.5~\mu\textup{m})^2$. The second dataset (see Figure \ref{fig:raw}) contains 5,373,442 positions recorded in 40,000 frames, of which 80,629 positions were assigned to two fiducial markers. As nearby fluorosphores cannot be discerned from the fiducials, this number slightly deviates from one registered position per frame per fiducial. The data of this set are distributed over an area of about $51~\mu\textup{m} \times 24~\mu\textup{m}$. The positions of the fiducials were used to compare the quality of our method to the current state of the art of drift correction.

To analyse the data with our method we create $T=2000$ position histograms of $n=N^2=256^2$ bins of the first data set and $T=2000$ position histograms of $n=N^2=512^2$ bins of the second data set, i.e. in both cases, we look at $T=2000$ position histograms which are composed of the data points of $T'/T = 20$ frames each (cf. Figure \ref{single.images:fig}). Note that, in particular, we made the positions histograms of the second dataset quadratic, compressing them in the $x_1$-direction. Our empirical analysis shows that the estimates are not strongly influenced by the choices of $T$ and $N$, however too small values circumvent the registration of small movements and for large values computational problems arise in terms of speed. This is in accordance with our previous simulation results.

\begin{figure}[htbp]
  \centering
\ifthenelse{\equal{\user}{alex}}{
  \includegraphics[width=15cm]{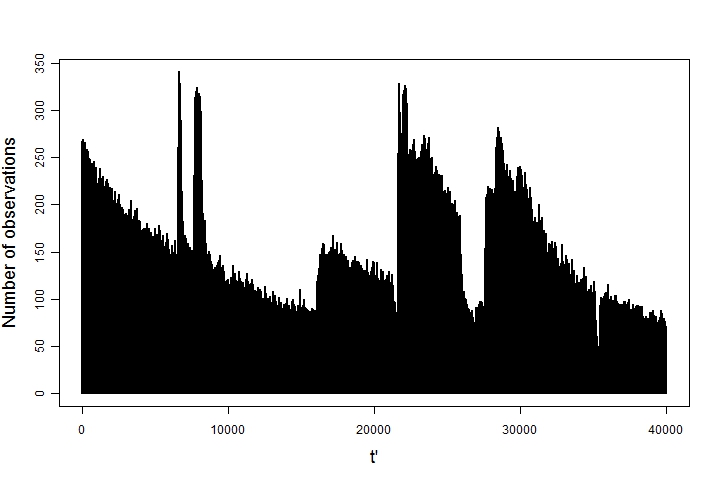}
  }{}
\ifthenelse{\equal{\user}{stephan}}{
  \includegraphics[width=15cm]{../IMG/Cage552_PVA_YG200_01_linear-quadratic_weights.jpg}
  }{}
  \caption{\it Number of data points (registered markers) $n'_{t'}$ per frame (network II)}
  \label{fig:weights}
\end{figure}

As exemplarily demonstrated in Figure \ref{fig:weights} for the second data set shown in Figure \ref{fig:raw}, the number of recorded markers $n'_{t'}$ varies as the experiment continues, possibly temporarily following a truncated exponential distribution.
Since switched on markers bleach after emitting light, one has to increase the switching laser intensity occasionally to get a roughly constant number of observations per frame. This is why the $n'_{t'}$ in Figure \ref{fig:weights} increase drastically every now and then.
The variation of the number of recorded positions implies that the uncertainty in the position histograms varies over time. Our method can easily account for this fact by maximizing a weighted version of $\tM_T(\vart)$,
\begin{eqnarray}\label{min2}
\tM_T^w(\vart) =
\sum_{k_1 =-\left\lfloor \xi_T/2\right\rfloor}^{\left\lfloor \xi_T/2\right\rfloor}
\sum_{k_2 =-\left\lfloor \xi_T/2\right\rfloor}^{\left\lfloor \xi_T/2\right\rfloor}
\left|\sum_{t=0}^{T-1} \omega_t\; \hat{Y}^{t}_{k_1,k_2} e^{2\pi i ( \tilde\delta_{t}(\vart)_1 k_1/N + \tilde\delta_{t}(\vart)_2 k_2/N)} \right|^2
\end{eqnarray}
with weights $\omega_t = n_t/\sum_t n_t$ for $t\in \{1,\dotsc,T\}$ and $\xi_T=0.2 N$.

\teacher{
Again the minimizer is evaluated by a standard Nelder-Mead-type algorithm with zero as starting value.
}

Note that experiments have been performed such that a fiducial marker has been included into the sample, i.e. a persistent fluorescence source, which enables us to track the drift easily, for testing purposes. We stress that this is currently state of the art technology to align SMS images over time (see the Introduction). In order to investigate the validity of our method, we delete the data originating from the fiducial markers from the observed sample and use it for verification only.
The visual inspection of the fiducial indicates a slight overall upward drift by 15 pixels i.e. about 715 nm (cf. Figure \ref{fig:raw}).

\begin{figure}[htbp]
  \centering
\ifthenelse{\equal{\user}{alex}}{
  \includegraphics[width=15cm]{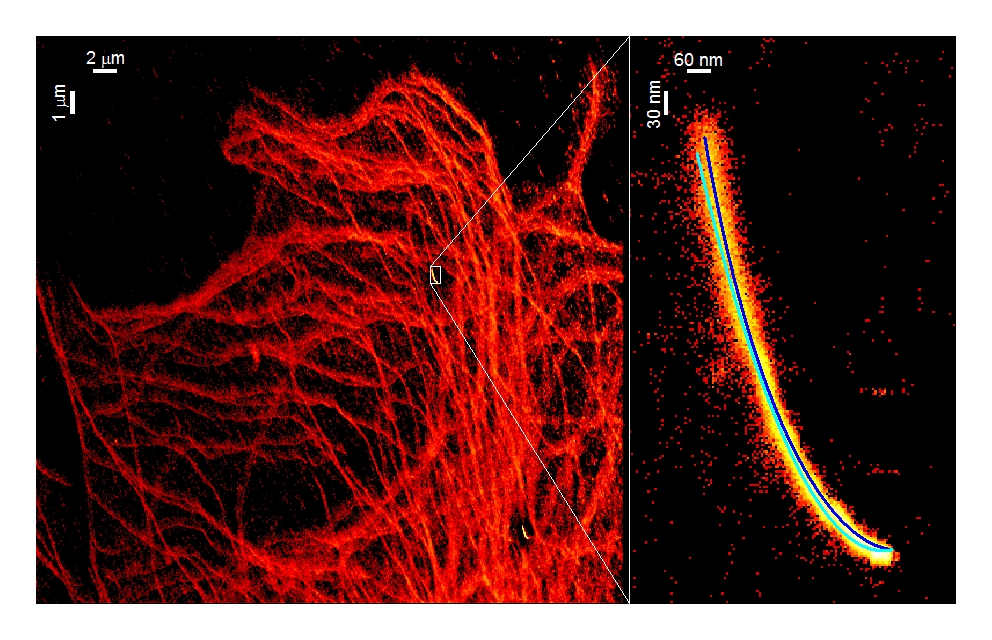}
  }{}
\ifthenelse{\equal{\user}{stephan}}{
  \includegraphics[width=15cm]{../IMG/Cage552_PVA_YG200_01_raw_image_and_bead.jpg}
  }{}
  \caption{\it Drift blurred position histogram of network II with $n=512^2$ bins (left) and close up on the area with the fiducial
  with a third order polynomial fit of its motion in blue and the estimated linear-quadratic drift curve in cyan (right).}
  \label{fig:raw}
\end{figure}

As reported in Tables \ref{tab:puredrift} and \ref{tab:app1}, to both datasets we apply four different drift models and choose the one with the smallest motion blur $m_2$ (cf. (\ref{blur.measure:eq})) to work with.
The time required for the computation of $\hvart_T$ depends on the considered drift type and may last up to several minutes (on a Core AMD Opteron with 2.6 GHz), depending on the bin width. If computational time is a major issue, we recommend for practical purposes to split the image in several domains and perform the drift estimation separately. The final estimator can be obtained by averaging.

Since we do not know the true drift function (as was the case in the simulation study), we determine the motion direction via minimization of $J$ (see Section \ref{sec:simulations} for details).
In particular in Table \ref{tab:app1} we report the $m_2$-value for the correction via fiducial tracking, too.
We track the fiducial marker by estimating its location at time $t \in \{1,\dotsc,T\}$ with the average of its data in the $t$-th position histogram.
The result indicates that our estimation method is at least competitive with tracking of the fiducial movement provided the motion is not severely misspecified (like linear/linear in Table \ref{tab:app1}). The reconstructions of the image for fiducial tracking and linear-quadratic drift are compared with one another in Figure \ref{fig:corr}.

\begin{table}[htbp]
\begin{center}
{\footnotesize
\begin{tabular}{ll||rrr|rrr|rr}
  \hline
\multicolumn{2}{c||}{drift models} & \multicolumn{3}{c|}{$x_1$-dir.} &\multicolumn{3}{c|}{$x_2$-dir.} & contrast & motion blur \\
$x_1$-dir.& $x_2$-dir.& $\hat\vart_{T;13}$ & $\hat\vart_{T;12}$ & $\hat\vart_{T;11}$ & $\hat\vart_{T;23}$ & $\hat\vart_{T;22}$ & $\hat\vart_{T;21}$ & $M_T$ & $m_2$ \\
  \hline
 linear & linear & - & - & -0.044 & - & - & 0.044 & 6.4679e-3 & 0.714 \\
 linear & quadratic & - & - & -0.047 & - & 0.059 & 0.002 & 6.4598e-3 & 0.582 \\
 quadratic & quadratic & - & -0.006 & -0.041 & - & 0.063 & -0.001 & 6.4594e-3 & 0.589 \\
 cubic & cubic & -0.002 & -0.002 & -0.043 & 0.051 & 0.003 & 0.015 & 6.4607e-3 & 0.593 \\
  \hline
 \multicolumn{2}{c||}{superimposed image} &\multicolumn{6}{c|}{~}& 6.4876e-3 & 0.830 \\ \hline
\end{tabular}
}
\caption{\it Estimation results for the  $\beta$-tubulin network I shown in Figure \ref{data.reconstr:fig} for several drift models.
}\label{tab:puredrift}
\end{center}
\end{table}

\begin{table}[htbp]
\begin{center}
{\footnotesize
\begin{tabular}{ll||rrr|rrr|rr}
  \hline
\multicolumn{2}{c||}{drift models} & \multicolumn{3}{c|}{$x_1$-dir.} &\multicolumn{3}{c|}{$x_2$-dir.} & contrast & motion blur \\
$x_1$-dir.& $x_2$-dir.& $\hat\vart_{T;13}$ & $\hat\vart_{T;12}$ & $\hat\vart_{T;11}$ & $\hat\vart_{T;23}$ & $\hat\vart_{T;22}$ & $\hat\vart_{T;21}$ & $M_T$ & $m_2$ \\
  \hline
 linear & linear & - & - & -0.009 & - & - & 0.014 & 8.7699e-3 & 0.546 \\
 linear & quadratic & - & - & -0.009 & - & 0.022 & -0.002 & 8.7443e-3 & 0.343 \\
 quadratic & quadratic & - & 0.001 & -0.01 & - & 0.022 & -0.002 & 8.7442e-3 & 0.344 \\
 cubic & cubic & 0.002 & -0.004 & -0.008 & 0.025 & -0.005 & 0.005 & 8.7464e-3 & 0.368 \\
  \hline
 \multicolumn{2}{c||}{fiducial tracking} &\multicolumn{6}{c|}{~}& 8.9203e-3 & 0.351 \\
 \multicolumn{2}{c||}{superimposed image} &\multicolumn{6}{c|}{~}& 8.8649e-3 & 0.897
\\ \hline
\end{tabular}
}
\caption{\it Estimation results for the $\beta$-tubulin network II with fiducial markers, cf. Figure \ref{fig:raw}, for several drift models. The displayed motion blur values are for the respective images with fiducial markers removed.}
\label{tab:app1}
\end{center}
\end{table}

\begin{figure}[htbp]
  \centering
\ifthenelse{\equal{\user}{alex}}{
  \includegraphics[width=\textwidth]{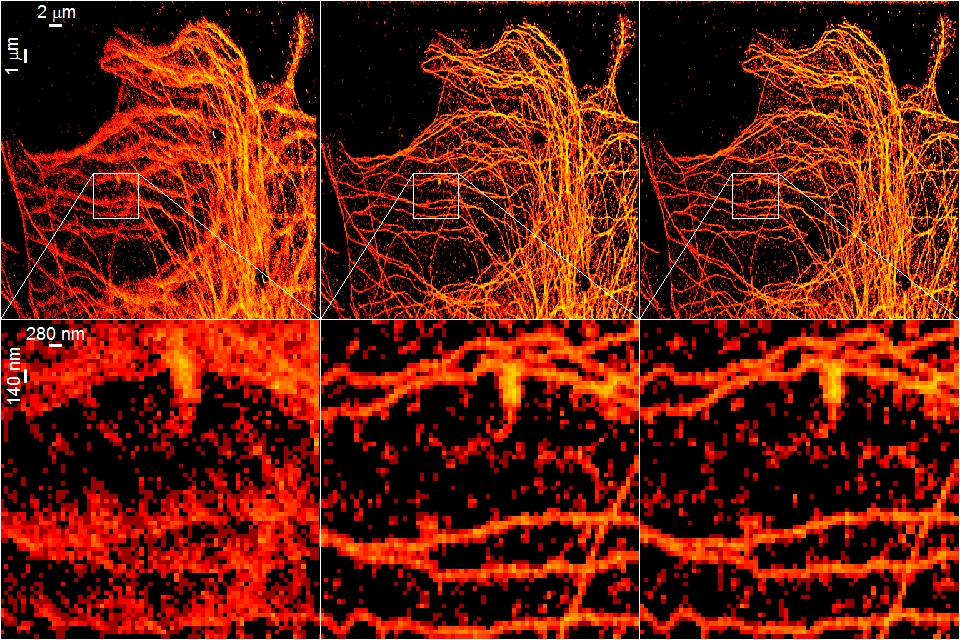}
  }{}
\ifthenelse{\equal{\user}{stephan}}{
  \includegraphics[width=\textwidth]{../IMG/Cage552_PVA_YG200_01_linear-quadratic_details}
  }{}
  \caption{\it Drift blurred network II (top left), by fiducial marker tracking corrected image (top center) and with the assumption of a linear-quadratic drift estimated image (top right), as well as detailed views inside the white boxes (bottom row)}
  \label{fig:corr}
\end{figure}

\section{Bootstrap confidence bands}\label{Bootstrap:scn}

\begin{figure}[htbp]
  \centering
\ifthenelse{\equal{\user}{alex}}{
  \includegraphics[width=15cm]{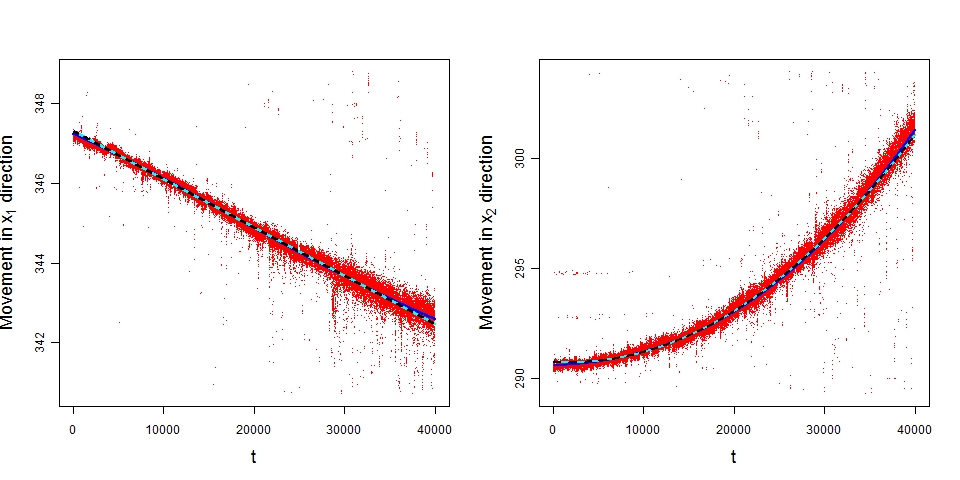}
  }{}
\ifthenelse{\equal{\user}{stephan}}{
  \includegraphics[width=15cm]{../IMG/Cage552_PVA_YG200_01_linear-quadratic_konfidenz.jpg}
  }{}
  \caption{\it Fiducial marker data in $x_1$- and $x_2$-direction with fitted third order polynomials (blue), estimated drift functions (cyan) and confidence bands (dashed). The movement axes are labelled in pixels, i.e. the fiducial data extends over an area of about $423\textup{ nm} \times 738\textup{ nm}$. Because we use the entire image, the confidence band is much sharper than the (few) fiducial marker data.}
  \label{fig:konf}
\end{figure}

\begin{figure}[htbp]
  \centering
\ifthenelse{\equal{\user}{alex}}{
  \includegraphics[width=7.5cm]{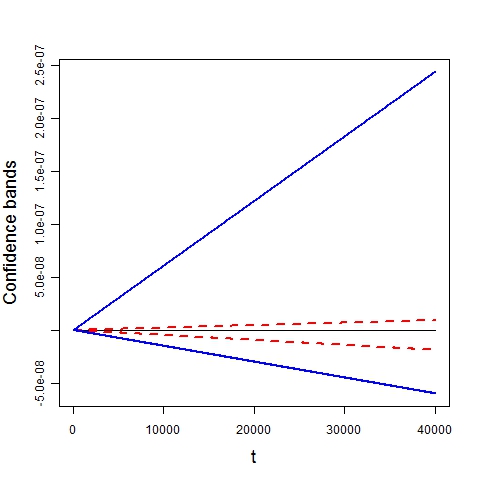}
  }{}
\ifthenelse{\equal{\user}{stephan}}{
  \includegraphics[width=7.5cm]{../IMG/Cage552_PVA_YG200_01_linear-quadratic_konfbaender_xy.jpg}
  }{}
  \caption{\it Bootstrap confidence bands for the drift functions in $x_1$-direction (dashed red) and $x_2$-direction (blue) with the respective estimated drift functions subtracted (cf. Figure \ref{fig:konf}). Because the distribution of the residuals has a high skewness of about 35 and because we chose confidence bands with minimized vertical width, the one in $x_2$-direction is highly non-symmetric.}
  \label{fig:konf2}
\end{figure}

\begin{figure}[htbp]
  \centering
\ifthenelse{\equal{\user}{alex}}{
  \includegraphics[width=7.5cm]{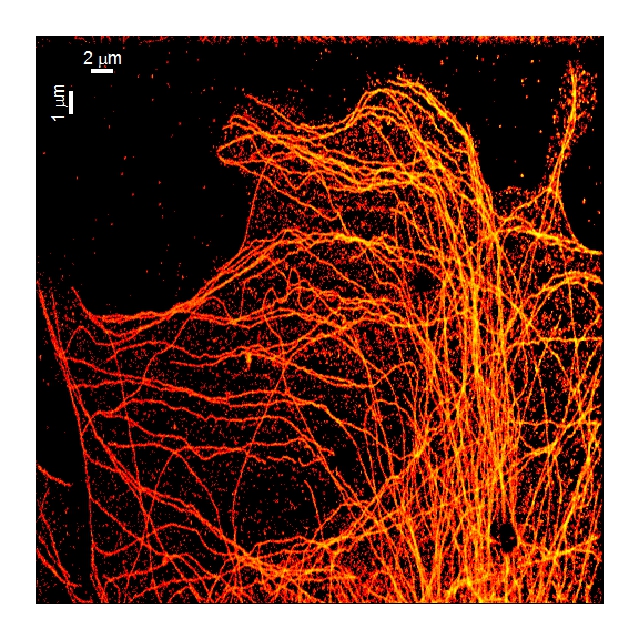}
  }{}
\ifthenelse{\equal{\user}{stephan}}{
  \includegraphics[width=7.5cm]{../IMG/Cage552_PVA_YG200_01_Bootstrapbild.jpg}
  }{}
  \caption{\it Average of the bootstrap replicates $\hat{f}_T^{(b_1)}, \dotsc, \hat{f}_T^{(b_m)}$ of the estimated image $\hat{f}_T$ corresponding to the $m = \lceil(1-\alpha)B\rceil$ drift curves $(t \mapsto \delta_t^{\hat{\vart}_T^{(b_j)}})_{j=1}^m$ nearest to the estimator $t \mapsto \delta_t^{\hvart_T}$ with respect to the supremum norm distance.}
  \label{fig:bootstrapbild}
\end{figure}

Given the estimators $\hat\vart_T$ and $\hat f_T$ from Definition \ref{estimator:def} and thus an estimator $\delta_t^{\hat\vart_T}$ for the drift function $\delta_t^{\vart_0}$, we can construct bootstrap confidence bands for the component functions $\delta_{t, i}^{\vart_0}$, $i \in \{1, 2\}$ using the method described in \cite{HallPittelkow} which we found particularly useful in our context. Here, we give a short summary of that method with our application to drift functions in mind. For notational simplicity following (\ref{few-observations:model}), we index the spatial location by a single index $j\in \{1,\dotsc, n\}$.

Note that this method assumes homoscedasticity (i.e. $\sigma^2_{j, t} \equiv \sigma^2 > 0$) and we use it for simplicity's sake. To account for heteroscedasticity, one could make use of a wild bootstrap procedure (see, e.g. \cite{Wu1986, Liu1988, Mammen1993}).

First, we consider the standardized difference $\Delta_t := (\delta_t^{\hat\vart_T}-\delta_t^{\vart_0})/\hat\sigma$, where $\hat\sigma$ is the empirical standard deviation of the \emph{residuals}
\begin{eqnarray}\label{residuals:eq}r_{j, t} &:=& Z_{j, t} - \hat f_T(x_j + \delta_t^{\hat\vart_T}),~ 1\leq j\leq n,0\leq t\leq T\end{eqnarray}
and thus an estimator for the standard deviation of the errors $\varep_{j, t} = Z_{j, t} - f(x_j + \delta_t(\vart_0))$. Obviously, constructing a confidence band for $\delta_t$ is equivalent to constructing one for $\Delta_t$.
Next we choose the shape of the confidence band in terms of two functions $g_+, g_- \colon [0, 1] \to [0, \infty)$ such that $\delta_t^{\hat\vart_T} + \hat\sigma u_+ g_+(t)$ and $\delta_t^{\hat\vart_T} - \hat\sigma u_- g_-(t)$ represent the upper and lower border, respectively, of the confidence band for $\delta_t^{\vart_0}$, with appropriate positive numbers $u_+, u_-$. For a confidence level $\alpha \in (0, 1)$ we minimize $u_++u_-$ under the constraint
\[ P\big(\delta_t^{\vart_0} \in [\delta_t^{\hat\vart_T} + \hat\sigma u_+ g_+(t), \delta_t^{\hat\vart_T} - \hat\sigma u_- g_-(t)] \text{ for all } t \in [0, 1]\big) \geq 1-\alpha, \]
or, equivalently under
\begin{equation*}
  P\big(\Delta_t \in [-u_+ g_+(t), u_- g_-(t)] \text{ for all } t \in [0, 1]\big) \geq 1-\alpha\,.
\end{equation*}
Since the distribution of $\Delta_t$ is unknown we approximate it by bootstrapping $B$ times from the residuals (\ref{residuals:eq}). For every $b \in \{1, \dotsc, B\}$ and every $1\leq j\leq n,0\leq t\leq T$ draw $\varep_{j, t}^{(b)}$ independently with replacement from the set of all residuals $\{r_{j',t'}: 1\leq j'\leq n, 0\leq t'\leq T\}$. Thus obtain
\[ Z_{j, t}^{(b)} := \hat f_T(x_j - \delta_t^{\hat\vart^{(b)}_T}) + \varep_{j, t}^{(b)}\,. \]
Applying our estimation method to the $Z_{j, t}^{(b)}$ we obtain bootstrap replicates $\hat\vart_T^{(b)}$ and thus replicates $\hat f_T^{(b)}$, $1\leq b\leq B$. This in turn leads to bootstrap replicates
\begin{eqnarray*}
r_{j, t}^{(b)} &:=& Z_{j, t}^{(b)} - \hat f_T^{(b)}(x_j + \delta_t^{\hat\vart_T^{(b)}})\,,\\
\hat\sigma^{(b)} &:=& \sqrt{\frac1{nT}\sum_{j, t}\left(r_{j, t}^{(b)} - \frac1{nT}\sum_{j', t'} r_{j', t'}^{(b)}\right)^2}\,,\\
\Delta_t^{(b)} &:=& (\delta_t^{\hat\vart_T^{(b)}}-\delta_T^{\hat\vart_T})/\hat\sigma^{(b)}
\end{eqnarray*}
which allow for minimization of $u_+ + u_-$ such that
\[ \#\big\{b \in \{1, \dotsc, B\}\mid\Delta_t^{(b)} \in [-u_+ g_+(t), u_- g_-(t)] \text{ for all } t \in [0, 1]\big\} \geq (1-\alpha)B. \]
Because we assume the drift to be zero at time $t=0$, we can employ a confidence band which has width zero at $t=0$. Since, in our application, we look at polynomial drift functions only, and since on $[0, 1]$ the linear part dominates the others in the sense that $t \geq t^p$ for all $p > 1$ and $t \in [0, 1]$, we will choose $g_+(t) = g_-(t) = t$.

The thus obtained confidence bands for the above data set of Figure \ref{fig:corr}, $B = 200$ and $\alpha = 0.05$ are shown in Figure \ref{fig:konf}, together with the data of the fiducial marker for comparison. For a better view of the very narrow confidence bands see Figure \ref{fig:konf2}, where we subtracted the respective estimated drift functions.

To further visualize the confidence statement we take a look at the average of (most of) the bootstrap replicats $\hat{f}_T^{(b)}$ of the estimator $\hat{f}_T$.
For that we choose the 0.95-proportion of the bootstrap replicates $\delta_t^{\hat{\vart}_T^{(b)}}$, $b \in \{1, \dotsc, B\}$, with the smallest supremum norm distances $\sup_{t \in [0, 1]} |\delta_t^{\hat{\vart}_T^{(b)}} - \delta_t^{\hvart_T}|$ to the original drift function estimator $\delta_t^{\hvart_T}$. We denote the corresponding indices with $b_1, \dotsc, b_m$, where $m = \lfloor (1-\alpha)B\rfloor$. The convex hull of the corresponding drift curves resembles a two-dimensional bootstrap confidence band. Figure \ref{fig:bootstrapbild} shows the average of the images $\hat{f}_T^{(b_1)}, \dotsc, \hat{f}_T^{(b_m)}$ which thus ``contains'' the true image with a probability of about 0.95. Remarkably, due to the small diameter of the confidence band (about three hundredths of the physical resolution of the data), this image is almost identical with the original estimator shown in Figure \ref{fig:corr}.

\section{Discussion and Outlook}

We proposed a method for drift estimation and correction in sparse dynamic imaging and derived its asymptotic distributional properties. On the one hand, sparse acquisition is beneficial for improved spatial resolution and an important feature of any SMS microscopy. On the other hand, we have seen that this provides a significant burden as it is well known that the specimens drift over time due to thermal inhomogeneity inside the sample and external systematic movements of the optical device. This raises a particular challenge for image registration as sparse acquisition and time drift provide a conflicting situation. Currently, this is solved by technically incorporating a bright fiducial marker into the specimen and registering its track (drift). We claim that this can be completely discarded in many applications and it is sufficient to apply the proposed statistical method to estimate the drift and finally to obtain the image from simply correcting the data by this drift. The proposed method has been 
investigated in simulations 
and in real world examples from SMS microscopy and for some examples even shown to outperform fiducial tracking. In general, reconstructions are quite satisfying.
In particular, the results show a certain degree of stability w{.}r{.}t{.} parameter choices, e.g. the threshold  $\xi_T$. Consistency and asymptotic normality of the proposed estimator has been established which allows to qualify the statistical error of the drift estimate and the final image. To this end, simple bootstrap methods can be used.

It remains to further work to investigate higher order properties of the proposed estimator as well as properties of nonparametric estimators, for example if $\delta_t$ is estimated by a spline. We believe that also semiparametric kernel based methods could be adapted to this problem as in single-index-modelling. Also an alteration of our method could be beneficial if one switches from the regression to a density viewpoint by looking at the Fourier transforms of the observed positions directly. Note that the proposed method can in principle be applied to higher dimensions, in particular three dimensional measurements. However, computationally this appears to be much more demanding.

\section*{Acknowledgements}
The authors acknowledge support from the Deutsche Forschungsgemeinschaft grant SFB 755 and the Volkswagen Foundation.
Stephan Huckemann acknowledges support from DFG HU 1275/2-1
and Axel Munk from DFG FOR 916. Finally, the authors would like to thank Timo Aspelmeier, Carsten Gottschlich, Thomas Hotz, and Yuri Golubev for helpful discussions.
\section{Appendix}

Recall the notations defined in Subsection 2.2.

\subsection{Proof of Theorem 2.9}

	\noindent{\it Plan of Proof.} We start with a proof of (9), which follows a standard three step argument in M-estimation (e.g. \cite{Vaart2000} and \cite{GamboaLoubesMaza2007}), although the details are quite elaborate. First we show the uniqueness of the population contrast minimizer $\vart_0$. In a second step we establish the continuity of $\vart\to\tM(\vart)$. Thirdly, we verify that $\tM_T(\vart) \to \tM(\vart)$ a.s. uniformly over $\vart\in\Theta$ as $T,\xi_T\to \infty$, $\xi_T=o(\sqrt{T})$. In consequence,  \cite[Theorem 5.7]{Vaart2000} (yielding weak consistency) can be adapted to obtain strong consistency. For convenience, here is the corresponding argument:

Since $\hat\vart_T$ is defined as a minimizer of $\tM_T$ (hence $\tM_T(\hat\vart_T) \leq \tM_T(\vart_0)$) and $\tM_T(\vart_0) \to \tM(\vart_0)$ a.s., we have a.s. that
\[ \mathop{\limsup}_{T\to\infty}\bigl(\tM_T(\hat\vart_T)- \tM(\vart_0)\bigr)
	= \mathop{\limsup}_{T\to\infty}\bigl(\tM_T(\hat\vart_T)- \tM_T(\vart_0)\bigr) + \lim_{T \to \infty} \bigl(\tM_T(\vart_0)-\tM(\vart_0)\bigr)
	\leq 0. \]
It follows that
	 \begin{eqnarray}
	  \mathop{\limsup}_{T\to\infty}\tM(\hat\vart_T) - \tM(\vart_0) &\leq & \mathop{\limsup}_{T\to\infty}\Big(\tM(\hat\vart_T) - \tM_T(\hat\vart_T)\Big)
	   \notag \\ \label{eq:PlanOfProofLimsup}
	  &\leq&\mathop{\limsup}_{T\to\infty} \sup_{\vart \in\Theta} \Big|\tM(\vart) - \tM_T(\vart)\Big|~=~0 \; \text{a.s.}
	 \end{eqnarray}
Because of the uniqueness of the minimizer $\vart_0$, the continuity of $\tM$ and the compactness of $\Theta$, we have that for every $\varep >0$ there is $\eta_\varep>0$ such that  $\tM(\vart)>\tM(\vart_0)+\eta_\varep$ for all $\vart\in\Theta$ with $\|\vart-\vart_0\|\geq \varep$. Hence
	 \begin{eqnarray*}
	P\Big(\mathop{\limsup}_{T\to\infty}\big\{\|\hat\vart_T -\vart_0\|\geq \varep\big\}\Big)
	&\leq & P\Big(\mathop{\limsup}_{T\to\infty}\big\{\tM(\hat\vart_T)>\tM(\vart_0)+\eta_\varep\big\}\Big)	 \\
	&\leq & P\Big\{\mathop{\limsup}_{T\to\infty}\tM(\hat\vart_T) \geq \tM(\vart_0)+\eta_\varep\Big\}~=~0\,,
      \end{eqnarray*}
where the last equality follows from (\ref{eq:PlanOfProofLimsup}).

	\paragraph{Step I: uniqueness of the contrast minimizer $\vart_0$.} First note that $\tM(\vart) \geq - \sum_{k \in \Z^2} |{f}_{k}|^2$ for all $\vart$ with equality for $\vart = \vart_0$. If this minimum is attained for some $\vart$ then for each $k$ with $|{f}_{k}|^2\ > 0$
	$$\left| \int_0^1 h_k(\delta^\vart_{t}-\delta^{\vart_0}_{t})\, dt \right|^2 = 1$$
	since $|\int_0^1h_k\,dt| \leq \int_0^1|h_k|\,dt=1$. This
	implies that $h_k(\delta^\vart_{t}-\delta^{\vart_0}_{t})=1$, i.e.
	$$2\pi \left\langle k , \delta_t^\vart-\delta_{t}^{\vart_0}\right\rangle \equiv 0 \,\textnormal{ mod }\, 2 \pi$$
	By Assumption 2.4 
	this holds for $k \in \left\{(k_1,k_2),(k'_1,k'_2)\right\}$ with $k_1k'_2-k_2k'_1\neq 0$.
	Hence, we can treat each dimension separately and obtain  $\delta_t^\vart\equiv \delta_{t}^{\vart_0}$ mod $2\pi$ a.e. Since this holds also for $k \in \left\{(k''_1,k''_2),(k'''_1,k'''_2)\right\}$ with $k''_1k'''_2-k''_2k'''_1\neq 0$, due to the part of the Assumption on non-common divisors we obtain  $\delta_t^\vart=\delta_{t}^{\vart_0}$ a.e. and hence $\vart=\vart_0$.

	\paragraph{Step II: continuity of $\tM$.} For $\vart,\vart'\in \Theta$ we have that
	\begin{eqnarray*}
	|\tM(\vart) - \tM(\vart')|  &\leq &  \sum_{k \in \Z^2} |{f}_{k}|^2 \left|
	\left| \sint h_k(\delta^\vart_{t}-\delta^{\vart_0}_{t})\, dt \right|^2
	- \left| \sint h_k(\delta^{\vart'}_{t}-\delta^{\vart_0}_{t})\, dt \right|^2 \right|\\
	&\leq& 2 \sum_{k \in \Z^2}  |{f}_{k}|^2  \left|
 	\sint \left(e^{2\pi i \left\langle k , \delta_t^\vart-\delta_{t}^{\vart_0}\right\rangle} - e^{2\pi i \left\langle k , \delta_t^{\vart'}-\delta_{t}^{\vart_0}\right\rangle}\right) dt \right| \\
	&\leq& 2 \sum_{k \in \Z^2}  |{f}_{k}|^2  \sint \left|1-e^{2\pi i\left\langle k , \delta_t^{\vart'}-\delta_{t}^{\vart}\right\rangle}\right| \,dt \\
	&\leq&  4\pi \sum_{k \in \Z^2} |k| |{f}_{k}|^2 \sint\left\| \delta_t^\vart-\delta_{t}^{\vart'} \right\| dt\,,
	\end{eqnarray*}
	where we use
	\begin{eqnarray}\label{eq:complex}
	|a|^2-|b|^2 \leq 2|a-b|
	\end{eqnarray}
	for $a,b \in \C$ with $|a|,|b| < 1$ in the second inequality and $|1-e^{ix}|^2=2-2\cos x \leq x^2$
	in the fourth one. By Assumptions 2.4, 2.6, 
	this implies the continuity of $\tM(\vart)$.

	\paragraph{Step III: $\tM_T \to \tM$ uniformly in $\vart$ a.s.}
	Recall from model (4) 
	that
	\[ Y_k^t = h_k(-\delta_t^{\vart_0}){f}_k + W^t_k \]
	with the true and unknown parameter $\vart_0\in \Theta$. Hence with (7) 
	we have that
	\begin{eqnarray*}
	\tM_T(\vart) &=&  - \sum_{|k| < \xi_T}
	\left|\frac{1}{T} \sum_{t\in\T}  \Big(h_k(\delta_t^\vart-\delta_t^{\vart_0}){f}_{k} + h_k(\delta_t^\vart) {W}^{t}_{k}\Big)\right|^2
	~=~
	A_T(\vart) - B_T(\vart) - C_T(\vart)
	\end{eqnarray*}
	with
	\begin{eqnarray*}
	A_T(\vart) &:=&  - \sum_{|k| < \xi_T}
	\left|\frac{1}{T} \sum_{t\in\T} h_k(\delta_t^\vart-\delta_t^{\vart_0}) {f}_{k} \right|^2, \\
	B_T(\vart) &:=& \sum_{|k| < \xi_T}  2 \RE \left(\Bigl( \frac{1}{T} \sum_{t\in\T} h_k(\delta_t^\vart-\delta_t^{\vart_0}) {f}_{k} \Bigr) \Bigl(\frac{1}{T} \sum_{t'\in\T}   h_k(-\delta_{t'}^\vart) \overline{{W}^{t'}_{k}}  \Bigr)\right), \\
	C_T(\vart) &:=& \sum_{|k| < \xi_T} \left|\frac{1}{T} \sum_{t\in\T}   h_k(\delta_t^\vart) {W}^{t}_{k}\right|^2.
	\end{eqnarray*}

	To derive the desired uniform convergence we will show for the deterministic part that $A_T \to \tM$ uniformly in $\vart$ while the random parts $B_T$ and $C_T$ converge to zero uniformly a.s.
	Considering
	\begin{eqnarray*}
	|A_T(\vart) - \tM(\vart)| &\leq& \sum_{|k| < \xi_T} |{f}_{k}|^2
	\left|\left|\frac{1}{T} \sum_{t\in\T} h_k(\delta_t^\vart-\delta_t^{\vart_0}) \right|^2  - \left| \sint h_k(\delta_t^\vart-\delta_t^{\vart_0})\, dt \right|^2\right|
	\\&&
	+  \sum_{|k| \geq \xi_T} |{f}_{k}|^2 \left| \sint h_k(\delta_t^\vart-\delta_t^{\vart_0})\, dt \right|^2\,, \end{eqnarray*}
	and applying (\ref{eq:complex}) again to the first sum while noting that the second is bounded by $\sum_{|k| \geq \xi_T} |{f}_{k}|^2 = o(1)$ ($\xi_T \to \infty$ by hypothesis and $\sum_k |{f}_{k}|^2 < \infty$ by Assumption 2.4) 
	gives
	\begin{eqnarray*}
	|A_T(\vart) - \tM(\vart)| &\leq&\sum_{|k| < \xi_T} 2 |{f}_{k}|^2
	\left|\frac{1}{T} \sum_{t\in\T} h_k(\delta_t^\vart-\delta_t^{\vart_0}) - \sint h_k(\delta_t^\vart-\delta_t^{\vart_0})\, dt \right| + o(1).
	\end{eqnarray*}
	Since the total variation of $t\mapsto h_k(\delta_t^\vart-\delta_t^{\vart_0})$ is bounded by a constant times $|k|$ uniformly in $\vart$ (Assumption 2.6), 
	we have for some constant $C$ that
	$$\left|\frac{1}{T} \sum_{t\in\T} h_k(\delta_t^\vart-\delta_t^{\vart_0}) - \sint h_k(\delta_t^\vart-\delta_t^{\vart_0})\, dt \right| < \frac{|k| \, C}{T}.
	$$
	In consequence of $\sum_k |k| |{f}_{k}|^2 < \infty$ (Assumption 2.4) 
	this implies that
	$$|A_T(\vart) - \tM(\vart)|= O(1/T)\,,$$
 	uniformly in $\vart$ as desired.
	Next, we show
	\begin{eqnarray}\label{supCTOxisqOverT:eq}
	\sup_{\vart\in\Theta} C_T(\vart) &=& \sup_{\vart\in\Theta} \sum_{|k| < \xi_T} \left|\frac{1}{T} \sum_{t\in\T}  h_k(\delta_t^\vart) {W}^{t}_{k}\right|^2  = o\left(\frac{\xi_T^2}{T}\right) \; \text{a.s.}
	\end{eqnarray}

	Since $h_k(\delta^\vart_t)$ acts as a rotation, $h_k(\delta^\vart_t)W_k^t =: U_k^t+iV_k^t$ ($t\in \T,|k| < \xi_T$) are again independently complex normally distributed; in particular, every $U_k^t=\RE(h_k(\delta^\vart_t)W_k^t)$ is independent of $V_k^t=\IM(h_k(\delta^\vart_t)W_k^t)$. Let
	\[ \bar U_{k, T} = \frac1{\sqrt{T}} \sum_{t \in \T} U_k^t, \quad \bar V_{k, T} = \frac1{\sqrt{T}} \sum_{t \in \T} V_k^t. \]
	Because of $E(\epsilon_{j,t}^4) = 3$ and Assumption 2.7 
	we have
	\begin{eqnarray*}
	  \var(\bar U_{k, T}^2) &\leq & E(\bar U_{k, T}^4) \\
	  &=& \frac3{T^2}\sum_{t\in\T}\frac1{n_t^2} \sum_{j \in J_t}\sigma_{j,t}^4\cos(-2\pi\langle k, x_{j,t}-\delta_t^\vart\rangle)^4 \\
	  && + \frac3{T^2}\sum_{t\neq t'}\frac1{n_t n_{t'}} \sum_{j \in J_t}\sum_{j' \in J_{t'}}\sigma_{j,t}^2\sigma_{j',t'}^2\cos(-2\pi\langle k, x_{j,t}-\delta_t^\vart\rangle)^2\cos(-2\pi\langle k, x_{j',t'}-\delta_{t'}^\vart\rangle)^2 \\
	  &\leq & 3\sigma_{\textup{max}}^4\left(\frac1{T^2}\sum_{t \in \T}\frac1{n_t} + 1\right) \leq 6 \sigma_{\textup{max}}^4,
	\end{eqnarray*}
	and similarly $\var(\bar V_{k, T}^2) \leq 6 \sigma_{\textup{max}}^4$. Again by Assumption 2.7,
	\begin{eqnarray*}
	  E(\bar U_{k, T}^2 + \bar V_{k, T}^2)
	  &=& \frac1{T}\sum_{t\in\T}\frac1{n_t} \sum_{j \in J_t}\sigma_{j,t}^2\left(\cos(-2\pi\langle k, x_{j,t}-\delta_t^\vart\rangle)^2 + \sin(-2\pi\langle k, x_{j,t}-\delta_t^\vart\rangle)^2\right) \\
	  &=& \frac1{T}\sum_{t\in\T}\frac1{n_t} \sum_{j \in J_t}\sigma_{j,t}^2 \leq \sigma_{\textup{max}}^2.
	\end{eqnarray*}
	In consequence, Kolmogorov's strong law (see e.g. \cite[Theorem 2.3.10]{Sen1993}) yields that
	\begin{eqnarray*}
	&& \left\vert \frac{1}{\#\left\{|k| < \xi_T\right\}}\sum_{|k|<\xi_T}\left|\frac{1}{\sqrt{T}}\sum_{t\in\T}h_k(\delta_t^\vart)W^t_k\right|^2 - \frac1{T}\sum_{t\in\T}\frac1{n_t} \sum_{j \in J_t}\sigma_{j,t}^2 \right\vert \\
	&= & \left\vert \frac{1}{\#\left\{|k| < \xi_T\right\}}\sum_{|k|<\xi_T}(\bar U_{k, T}^2 + \bar V_{k, T}^2) - \frac1{T}\sum_{t\in\T}\frac1{n_t} \sum_{j \in J_t}\sigma_{j,t}^2 \right\vert \\
	&\to & 0 \; \text{a.s.}, \quad T \to \infty.
	\end{eqnarray*}
	Since $\#\left\{|k| < \xi_T\right\} = O(\xi_T^2)$ this yields (\ref{supCTOxisqOverT:eq}).
	Finally,
	$$ \sup_{\vart } |B_T(\vart)|^2 = o(1) \; \text{a.s.}$$
	follows at once from  $ |A_T(\vart)|\leq \sum_k |{f}_{k}|^2$ by definition, (\ref{supCTOxisqOverT:eq}) and the observation that
 	$|B_T(\vart)|^2 \leq  2 | A_T(\vart)| \,|C_T(\vart)|$. This concludes the proof of Step III.

	\paragraph{The proof of (10).}
	Observe that, using the Plancherel equality, we have
	\begin{eqnarray} \nonumber
	\left\| \hat f_T - f \right\|_2^2 &=& \sum_{|k| < \xi_T} \left|\frac{1}{T} \sum_{t\in\T}   h_k(\delta_t^{\hvart_T}) {Y}^{t}_{k} - {f}_k \right|^2 +
	\sum_{|k| \geq \xi_T} \left| {f}_k \right|^2
	\\  \nonumber &=& 
	\sum_{|k| < \xi_T} \left|\frac{1}{T} \sum_{t\in\T} \big( h_k(\delta_t^{\hvart_T}-\delta_t^{\vart_0}) {f}_k + h_k(\delta_t^{\hvart_T}) {W}^{t}_{k}\big) - {f}_k \right|^2
	+ o(1)
	\end{eqnarray}
	\begin{eqnarray} \nonumber
	&=&
 	\sum_{|k| < \xi_T} |{f}_k|^2\frac{1}{T^2}\sum_{t,t'\in\T} \big(h_k(\delta_t^{\hvart_T}-\delta_t^{\vart_0}) - 1\big)\big(h_k(-\delta_{t'}^{\hvart_T}+\delta_{t'}^{\vart_0}) - 1\big)
	\\  \nonumber &&+ \, \sum_{|k| < \xi_T} \left| \frac{1}{T} \sum_{t\in\T} h_k(\delta_t^{\hvart_T})  {W}^{t}_{k} \right|^2
	\\  \nonumber &&~~ + 2\sum_{|k| < \xi_T}\frac{1}{T^2}\sum_{t,t'\in\T}\big(h_k(\delta_t^{\hvart_T}-\delta_t^{\vart_0})-1\big) {f}_k
	h_k(-\delta_{t'}^{\hvart_T}) \overline{{W}^{t'}_{k}}
 	+ o(1)
	\\  \label{eq:thmb:proof} &\leq& 
 	4\pi L\|\hvart_T-\vart_0\|\,\sum_{|k| < \xi_T} \left(
	|{f}_k|^2 \,|k| + |{f}_k|\,|k|\,\frac{1}{\sqrt{T}}  |{G}^{T}_{k}|\right) + o(1) \; \text{a.s.}
	\end{eqnarray}
	with $G_k^T$ defined below, by (\ref{supCTOxisqOverT:eq}),
	since $|h_k(\delta_{t}^{\hvart_T}-\delta_{t}^{\vart_0}) - 1|\leq 2$ as well as (recalling the argument following display (\ref{eq:complex}))
	$$\Big|h_k(\delta_{t}^{\hvart_T}-\delta_{t}^{\vart_0}) - 1\Big| \leq 2\pi |k|\|\delta_{t}^{\hvart_T}-\delta_{t}^{\vart_0}\| \leq 2\pi L |k| \|\hvart_T - \vart_0\|$$
	with the constant $L>0$ from Assumption 2.8 
	and the following argument. Setting
	$$G^T_k := \frac{1}{\sqrt{T}}\sum_{t'\in\T}h_k(-\delta_{t'}^{\hvart_T}) \overline{{W}^{t'}_{k}}\,,$$
	we obtain complex normal deviates independent in $k$ with the property
      \begin{eqnarray*}\frac{1}{T^2}\sum_{t,t'\in\T}\big(h_k(\delta_t^{\hvart_T}-\delta_t^{\vart_0})-1\big) {f}_k
	h_k(-\delta_{t'}^{\hvart_T}) \overline{{W}^{t'}_{k}}&=& \frac{f_k}{\sqrt{T}}\left(\frac{1}{T}\sum_{t\in\T}\big(h_k(\delta_t^{\hvart_T}-\delta_t^{\vart_0})-1\big)\right) G_k^T.
      \end{eqnarray*}
	Now (\ref{eq:thmb:proof}) yields indeed $\|\hat f_T - f \|_2^2 \to 0$ a.s. if $\xi^2_T/\sqrt{T} \to 0$ since
	$\|\hvart_T-\vart_0\|\to 0$ a.s. as shown in the proof of the first part of Theorem 2.9, 
	$\sup_{k\in \mathbb Z}|{f}_k|\,|k|<\infty\,$ by Remark 2.5 
	and $\sum_{|k| < \xi_T} |{f}_k|^2 \,|k|< \infty\,$
	by Assumption 2.4. 
	The same argument that led to (\ref{supCTOxisqOverT:eq}) shows that the variance of
	\[ \frac{1}{\sqrt{T}}\sum_{|k| < \xi_T} |{f}_k|\,|k|\,  |{G}^{T}_{k}| \]
	is of order $o(1)$ in case of $\xi_T/\sqrt{T} \to 0$, which gives convergence of $\|\hat f_T-f\|_2\to 0$ in probability, completing the proof.
	\qed

\subsection{Proof of (i) of Theorem 2.13}

	With the $d$-dimensional real vector $\ba_{k,t}^\vart := 2\pi \grad_\vart \langle k,\delta_t^\vart\rangle$ verify that
	\begin{eqnarray}\nonumber\label{gradYMproof1:eq}
	 \grad_\vart\left(\sum_{t\in\T}h_k(\delta_t^\vart)Y_k^t\,\sum_{t'\in\T}\overline{h_k(\delta_{t'}^\vart)Y_k^{t'}}\right)
	  &=&
		 2\RE\left(\sum_{t,t'\in\T}\grad_\vart\Big(h_k(\delta_t^\vart)Y_k^t\Big)\,\overline{h_k(\delta_{t'}^\vart)Y_k^{t'}}\right)\\
	  &=&-2\IM\left(\sum_{t,t'\in\T}\ba_{k,t}^\vart h_k(\delta_t^\vart)Y_k^t \,\overline{h_k(\delta_{t'}^\vart)Y_k^{t'}}\right)\,.
	\end{eqnarray}
	Moreover, with the true parameter $\vart_0\in \Theta$ and arbitray $\vart\in\Theta$ recall from (2) 
	that
	$$h_k(\delta_t^\vart)Y_k^t = h_k(\delta_t^\vart-\delta_t^{\vart_0}) f_k + h_k(\delta_t^\vart)W_k^t\,.$$
	At $\vart=\vart_0$ the right hand side is just $f_k + h_k(\delta_t^{\vart_0})W_k^t$.
	In consequence we have for $\tM_T$ from (7) 
	that
	\begin{eqnarray}\label{gradMTsumHTk:eq}
	 \grad_\vart \tM_T(\vart_0) &=& \sum_{|k|\leq \xi_T} H^T_k
	\end{eqnarray}
	where $\ba_k^t = \ba_{k,t}^{\vart_0}$, $f_k = e_k+ig_k$, $h_k(\delta_t^{\vart_0})W_k^t = \tau_k^t A_k^t+i \omega_k^t B_k^t$ with standard deviations
	\begin{align*}
	  \tau_k^t &:= \sqrt{\frac1{n_t}\sum_{j \in J_t} \sigma_{j, t}^2 \cos(-2\pi\langle k, x_{j,t}-\delta_t^{\vart_0}\rangle)^2}, \\
	  \omega_k^t &:= \sqrt{\frac1{n_t}\sum_{j \in J_t} \sigma_{j, t}^2 \sin(-2\pi\langle k, x_{j,t}-\delta_t^{\vart_0}\rangle)^2},
	\end{align*}
	and
	\begin{eqnarray*}
	  H^T_k &:=& \frac{2}{T^2}\IM\left(\sum_{t,t'\in\T} \ba_k^t\Big(|f_k|^2 + f_k \overline{h_k(\delta_{t'}^{\vart_0})W_k^{t'}}+ h_k(\delta_t^{\vart_0})W_k^{t} \overline{f_k}+ h_k(\delta_t^{\vart_0})W_k^t\,\overline{h_k(\delta_{t'}^{\vart_0})W_k^{t'}}\Big)\right)\\
	  &=& \frac{2}{T^2}\sum_{t,t'\in\T}\ba_k^t \Big(g_k \tau_k^{t'} A_k^{t'} - e_k \omega_k^{t'} B_k^{t'} +e_k \omega_k^t B_k^t-g_k \tau_k^t A_k^t + \tau_k^{t'} \omega_k^t A_k^{t'}B_k^t - \tau_k^t \omega_k^{t'} A_k^tB_k^{t'} \Big)\,.
	\end{eqnarray*}
	Note that $A^t_k,B^t_k \sim {\cal N}(0,1)$ ($k \in \Z^2$, $t \in \T$) are all mutually independent, and for $k=(0, 0)$ we have $\omega_{(0, 0)}^t \equiv 0$.
	
	To determine the limit distribution of $\sqrt{T}\grad_\vart M_T(\vart)$ we look at its projections $\sqrt{T} \langle x,\grad_\vart M_T(\vart)\rangle$ with arbitrary but fixed $0\neq x=(x_1,\dotsc,x_d)\in \mathbb R^d$.
	To this end denote by $H_k^T(j)$ and $\ba_k^t(j)$ the $j$-th component of $H_k^T$ and $\ba_k^t$, respectively, $j\in\{1,\dotsc,d\}$, and set
	\begin{eqnarray}\label{GTk-def:eq} G_k^T := \sum_{j=1}^d x_j H_k^T(j),~~a_k^t :=  \sum_{j=1}^d x_j\ba_k^t(j)\,.\end{eqnarray}
	Introducing the independent normal vectors $A_k:=(\tau_k^t A_k^t/\bar\tau_k^T)_{t\in\T}$, $B_k:=(\omega_k^t B_k^t/\bar\omega_k^T)_{t\in\T}$ with (cf. Assumption 2.12)
	\[ \bar\tau_k^T = \sqrt{\frac1T\sum_{s \in \T}(\tau_k^s)^2} > 0, \quad
	   \bar\omega_k^T = \sqrt{\frac1T\sum_{s \in \T}(\omega_k^s)^2} > 0, \]
	each with independent components as well as the unit vector $e := (1)_{t\in\T}/\sqrt{T}$ and the vector $a_k = (a_k^t)_{t\in\T}$ and denoting the transpose of $a_k$ by $a'_k$ etc., we obtain
	\begin{eqnarray*}
	  G_k^T &=&\frac{2\bar\tau_k^T\bar\omega_k^T}{T^{3/2}}\Big( a'_kB_k A'_k e -  e' B_kA'_ka_k\Big)\\
	  && +\frac{2}{T}\Big(\bar\tau_k^T g_ka'_kee'A_k -\bar\omega_k^T e_ka'_kee'B_k+\bar\omega_k^T e_ka'_kB_k-\bar\tau_k^T g_ka'_kA_k\Big)\,.
	\end{eqnarray*}
	To tackle the first term introduce a unit vector $b_k$ orthogonal to $e$ such that $a_k =\alpha_k e + \beta_k b_k$, $\alpha_k,\beta_k \in \mathbb R$ and define a matrix $U=U_k \in SO(T)$ having $e$ and $b_k$ as the first two columns. Then, with the independent normal vectors $\wA_k = U'A_k$, $\wB_k = U'B_k$ with independent components, each with zero mean,
	\begin{eqnarray*}  a'_kB_k A'_k e -  e' B_kA'_ka_k &=& A'_k(ea'_k-a_ke')B_k\\
	&=&A'_kUU'( e a'_k - a_k e') UU'B_k  \\
	&=& A'_kU(e,b_k,*)' \Big(e (\alpha_k e+ \beta_k b_k)' -(\alpha_k e+ \beta_k b_k)e'\Big) (e,b_k,*)U'B_k \\
	&=& \wA'_k\Big( (1,0,\dotsc,0)' (\alpha_k,\beta_k,0,\dotsc,0) - (\alpha_k,\beta_k,0,\dotsc,0)' (1,0,\dotsc,0)\Big)\wB_k \\
	&=&\wA'_k\,\beta_k\left(\begin{array}{ccccc}0&1&0&\cdots&0\\
	          -1&0&0&\cdots&0\\
		  0&0&0&\cdots&0\\
	     \vdots&\vdots&\vdots&\ddots& \vdots\\
		  0&0&0&\cdots&0
	         \end{array}\right)\wB_k\,.
	\end{eqnarray*}
	In consequence, with the first components $\wA_k^{(1)}$, $\wB_k^{(1)}$ and second components $\wA_k^{(2)}$, $\wB_k^{(2)}$ of $\wA_k$ and $\wB_k$,
	\begin{eqnarray*}
	G_k^T&=& \frac{2\bar\tau_k^T\bar\omega_k^T\beta_k}{T^{3/2}} \Big(\wA_k^{(1)}\wB_k^{(2)}-\wA_k^{(2)}\wB_k^{(1)}\Big) \\
	&& + \frac{2}{T}\Big(\bar\tau_k^Tg_k\alpha_k\wA_k^{(1)}-\bar\omega_k^Te_k\alpha_k\wB_k^{(1)} + \bar\omega_k^Te_k(\alpha_k \wB_k^{(1)} +\beta_k\wB_k^{(2)}) -\bar\tau_k^Tg_k(\alpha_k\wA_k^{(1)}+ \beta_k \wA_k^{(2)})\Big)\,.
	\end{eqnarray*}
	At this point we note that
	\begin{eqnarray}\label{betaksq:eq}\beta_k^2&=&\|a_k - \alpha_k e\|^2 =\sum_{t\in\T} \left( a_k^t-\frac{1}{T}\sum_{t'\in\T} a_k^{t'}\right)^2
	  ~=~
	\sum_{t\in\T}(a_k^t)^2 - \frac{1}{T}\left(\sum_{t\in\T}a_k^t\right)^2\,
	\end{eqnarray}
	whence $\beta_k =O(|k|\sqrt{T})$ from the definition of $a_k^t$ and Assumption 2.11. 
	Furthermore, by Assumption 2.12, $\bar\tau_k^T \to \sigma_{A, k}$ and $\bar\omega_k^T \to \sigma_{B, k}$ uniformly in $k$ as $T \to \infty$.
	Hence, the variance of the first term of $G_k^T$  scales with $|k|^2/T^2$, thus
	\begin{equation}\label{tight-limit:eq} (G^T)_1 := \sum_{|k|\leq \xi_T}\frac{2\bar\tau_k^T\bar\omega_k^T\beta_k}{T^{3/2}} \Big(\wB_k^{(1)}\wA_k^{(2)}-\wB_k^{(2)}\wA_k^{(1)}\Big)= O_p\left(\sqrt{\sum_{|k|< \xi_T} \frac{|k|^2}{T^2}}\right)=  O_p(\xi_T^{2}/T)\,
	\end{equation}
	i.e. with the hypothesis $\xi^4_T/T\to 0$, we obtain
	\begin{eqnarray}\label{CLT-zero-part:eq}\sqrt{T}\,(G^T)_1 ~\to~ 0\mbox{ in probability.}\end{eqnarray}
	Let us further note at this point for future use in case of $\xi_T\to\infty$ with $\xi^4_T/T\to 0$ due to $\beta_k \leq C|k|\sqrt{T}$ with a suitable constant $C>0$, we have also
	that
	\begin{eqnarray}\label{CLT-zero-part-as:eq}
	 |(G^T)_1|&\leq& \xi_T^2 \,\frac{1}{\xi_T^2} \sum_{|k|< \xi_T}\frac{2\bar\tau_k^T\bar\omega_k^TC\xi_T}{T}\Big|\wB_k^{(1)}\wA_k^{(2)}-\wB_k^{(2)}\wA_k^{(1)}\Big|~\to~0 \; \text{a.s.}
	\end{eqnarray}

	The second term of $G_k^T$ reduces to
	\[ \frac{2}{T}\Big(\bar\omega_k^Te_k\beta_k \wB_k^{(2)}  - \bar\tau_k^Tg_k\beta_k \wA_k^{(2)}\Big)\, \]
	which is normally distributed with zero mean and variance
	\begin{eqnarray*}
	 \lefteqn{\frac{4}{T^2}\beta^2_k\bigl((\bar\tau_k^Tg_k)^2+(\bar\omega_k^Te_k)^2\bigr)}\\
	 &=& \frac{16\pi^2\bigl((\bar\tau_k^Tg_k)^2+(\bar\omega_k^Te_k)^2\bigr)}{T^2}\,\left[\sum_{t\in\T}\left\langle k,\sum_{j=1}^dx_j\partial_{\vart_j}\delta_t^\vart\right\rangle^2-\frac{1}{T}\left(\sum_{t\in\T}\left\langle k,\sum_{j=1}^dx_j\partial_{\vart_j}\delta_t^\vart)\right\rangle\right)^2\right]\,,
      \end{eqnarray*}
	for $\vart=\vart_0$ cf. (\ref{betaksq:eq}).
	Since the normal random deviates in
	\[ (G^T)_2 := \sum_{|k|< \xi_T}\frac{2}{T}\Big(\bar\omega_k^Te_k\beta_k \wB_k^{(2)}-\bar\tau_k^Tg_k\beta_k \wA_k^{(2)}\Big) \]
	 are independent in $k$, we have that $\sqrt{T}\,(G^T)_2$ is normally distributed with zero mean and variance converging to
	\begin{align}\label{CLT-variance:eq}
	\notag &16\pi^2\sum_{k\in \Z^2}\bigl((\sigma_{A, k}g_k)^2+(\sigma_{B, k}e_k)^2\bigr)\left[\int_0^1\Big\langle k,(\grad_\vart \delta_t^{\vart_0})' x\Big\rangle^2\,dt-\left\langle k, \int_0^1(\grad_\vart \delta_t^{\vart_0})'x\,dt\right\rangle^2\right]\\
	&=:\sigma^2_{x} < \infty
	\end{align}
	if $f\in H^1\big([0,1]\big)$.
	Recalling the notation of (\ref{gradMTsumHTk:eq}), (\ref{GTk-def:eq}) and $\sum_{|k|<\xi_T} G_k^T=(G^T)_1+(G^T)_2=\langle x,\grad_\vart M_T(\vart)\rangle$ as well as collecting the results of (\ref{CLT-zero-part:eq}) and (\ref{CLT-variance:eq}) we have thus shown that for any $0\neq x\in \mathbb R^d$
	$$ \sqrt{T} \langle x,\grad_\vart M_T(\vart)\rangle \to {\cal N}(0,\sigma_x^2)$$
	whenever $T,\xi_T\to \infty$ with $\xi_T$ of rate $o(T^{1/4})$.
	Since this holds true for every $x$, the joint distribution of $\sqrt{T}\grad_\vart M_T(\vart)$ at $\vart=\vart_0$ is asymptotically multivariate normal with covariance matrix as asserted in Theorem 2.13.

	In view of use below we note here that 
	we obtain with suitable constants $C,C'>0$ ($C'$ due to Remark 2.5), 
	$\sigma_{\textup{max}}$ from Assumption 2.7 and independent standard normal $C_k$ ($k\in\mathbb Z$) that
	\begin{eqnarray}\nonumber\label{CLT-normal-part-as:eq}
	 |(G^T)_2|&=&\left| \frac{2}{T}  \sum_{|k|< \xi_T}\beta_k \sqrt{(\bar\tau_k^Tg_k)^2+(\bar\omega_k^Te_k)^2} C_k\right|
	~\leq~
	\frac{2\sigma_{\textup{max}} C}{\sqrt{T}}  \sum_{|k|< \xi_T}|f_k||k||C_k|
	\\
		&\leq &
	\frac{2\sigma_{\textup{max}} CC'\xi^2_T}{\sqrt{T}} \frac{1}{\xi_T^2} \sum_{|k|< \xi_T}|C_k|
	\to 0~ \; \text{a.s.}\mbox{ if $\xi_T \to \infty$ and $\xi^4_T/T = O(1)$}.
	\end{eqnarray}

	\qed

\begin{Rm}
As shown above, asymptotic normality of the second part $\sqrt{T}\,(G^T)_2$ of $\sqrt{T}\,\grad_\vart\tM_T(\vart_0)$ holds regardless of the rate of $\xi_T$.
	If we relax $\xi^4_T/T\to 0$ to  $C_1T^{1/4}\leq \xi_T \leq C_2 T^{1/4}$ with suitable constants $C_1,C_2>0$, the first part $\sqrt{T}\,(G^T)_1$ will no longer converge to zero but will be tight, cf. (\ref{tight-limit:eq}). Since then also $\hat\vart \to \vart_0$ by Theorem 2.9, 
	although the $(G^T)_1$ and $(G^T)_2$ will be dependent for this rate of $\xi_T$, we expect that asymptotic normality still holds. The corresponding covariance matrix, however, will have a more complicated structure than being a multiple of $\tilde\Sigma$.

\teacher{
	\item In a case of $\xi^T = T^{1/2-\varep}$, $0<\varep < 1/4$ we still have convergence $\hat\vart \to \vart_0$ a.s. but possibly with a non-tight limiting distribution.
}

\end{Rm}

\subsection{Proof of (ii) of Theorem 2.13}

	Here we build on the proof of (i) of Theorem 2.13 
	within the preceding section and use the notation there. In addition let $\bb_{k,t}^\vart:=2\pi{\Hess}_\vart \langle k,\delta_t^\vart\rangle$. Then we obtain at once from (\ref{gradYMproof1:eq})
	\begin{eqnarray}\nonumber
	{\rm Hess}_\vart\left(\sum_{t\in\T}h_k(\delta_t^\vart)Y_k^t\,\sum_{t'\in\T}\overline{h_k(\delta_{t'}^\vart)Y_k^{t'}}\right)
      &=&D_k^T + F_k^T
      \end{eqnarray}
      with
      	\begin{eqnarray*}
      D_k^T &:=&-2\IM
		\left(\sum_{t,t'\in\T}\bb_{k,t}^\vart
      h_k(\delta_t^\vart)Y_k^t\,\overline{h_k(\delta_{t'}^\vart)Y_k^{t'}}\right)
      \\
      F_k^T &:=&
      -2\RE\left(\sum_{t,t'\in\T}\ba_{k,t}^\vart(\ba_{k,t}^\vart - \ba_{k,t'}^\vart)'h_k(\delta_t^\vart)Y_k^t\,\overline{h_k(\delta_{t'}^\vart)Y_k^{t'}}\right)\,.
	\end{eqnarray*}
      In particular, in consequence of (7)
      \begin{eqnarray}\label{HessMT-Decomp.eq}
      {\rm Hess}_\vart \tM_T(\vart)
      &=& -\,\frac{1}{T^2}\sum_{|k|< \xi_T}  (D_k^T + F_k^T) \,.
      \end{eqnarray}
      Note that $E(D^T_k)=0$. Setting $\vart=\vart_0$ observe that the argument of the previous section (using the matrices $\bb_{k,t}^\vart$ instead of the vectors $\ba_{k,t}^\vart$) that led to (\ref{CLT-zero-part-as:eq}) and (\ref{CLT-normal-part-as:eq}) gives at once
      \begin{eqnarray}\label{HessMT-D-conv.eq}
      \frac{1}{T^2}\sum_{|k|< \xi_T}  D_k^T &\to& 0 \; \text{a.s.} \mbox{ if } T,\xi_T\to\infty\mbox{ and }\xi^4_T/T \to 0\,.
      \end{eqnarray}
      Likewise, the same follows for the random part of $F_k^T$.
      More precisely for $\vart=\vart_0$:
      \begin{eqnarray*}
       F_k^T&=&-2\sum_{t,t'\in\T}\ba_{k,t}^{\vart_0} (\ba_{k,t}^{\vart_0} - \ba_{k,t'}^{\vart_0})'\\&&~~
      \RE\Big(|f_k|^2 + f_k \overline{h_k(\delta_{t'}^{\vart_0})W_k^{t'}}+ h_k(\delta_t^{\vart_0})W_k^{t} \overline{f_k}+ h_k(\delta_t^{\vart_0})W_k^t\,\overline{h_k(\delta_{t'}^{\vart_0})W_k^{t'}}\Big)\\
      &=&-2\sum_{t,t'\in\T}|f_k|^2\ba_{k,t}^{\vart_0}(\ba_{k,t}^{\vart_0} - \ba_{k,t'}^{\vart_0})' + \tF_k^T
      \end{eqnarray*}
      with
      \begin{eqnarray*}
       \lefteqn{\tF_k^T:}\\&=& -2\sum_{t,t'\in\T}\ba_{k,t}^{\vart_0} (\ba_{k,t}^{\vart_0}- \ba_{k,t'}^{\vart_0})'\RE\Big(f_k \overline{h_k(\delta_{t'}^{\vart_0})W_k^{t'}}+ h_k(\delta_t^{\vart_0})W_k^{t} \overline{f_k}+ h_k(\delta_t^{\vart_0})W_k^t\,\overline{h_k(\delta_{t'}^{\vart_0})W_k^{t'}}\Big)
      \end{eqnarray*}
      yields
      \begin{eqnarray}\label{HessMT-tF-conv.eq}
       E(\tF^T_k)~=~0\mbox{ and }\frac{1}{T^2}\sum_{|k|< \xi_T}  \tF_k^T &\to& 0 \; \text{a.s.}\mbox{ if } T,\xi_T\to\infty\mbox{ and }\xi^4_T/T \to \infty\,.
      \end{eqnarray}
      Since we have the deterministic limit  	
      \begin{eqnarray*}
       \sum_{|k|< \xi_T}\frac{2}{T^2}\sum_{t,t'\in\T}|f_k|^2\ba_{k,t}^{\vart_0}(\ba_{k,t}^{\vart_0} - \ba_{k,t'}^{\vart_0})'&\to&2\sum_{k\in\mathbb Z^2}|f_k|^2\dint_{[0,1]^2} \ba_{k,t}^{\vart_0}(\ba_{k,t}^{\vart_0} - \ba_{k,t'}^{\vart_0})'\,dtdt'
      \end{eqnarray*}
       as $T,\xi_T\to \infty$ due to Assumption 2.11 
       on bounded total variation of first $\vart$-derivatives, in conjunction with (\ref{HessMT-Decomp.eq}), (\ref{HessMT-D-conv.eq}) and (\ref{HessMT-tF-conv.eq}) the definition of $\ba^{\vart_0}_{k,t}$ yields the assertion (ii) of Theorem 2.13.

\subsection{Ad Example 2.15}

\begin{Lem}
       In the situation of Example 2.15, $\det(\Sigma) = 0$ iff there is $x \in \R^2\setminus\{0\}$ s.t.
       \begin{equation}
          \label{eq:fconstant}
          f(y+rx)=f(y) \quad \text{for all } y \in \R^2, r \in \R,
       \end{equation}
       where $f$ is $[0, 1]^2$-periodic.
\end{Lem}

\begin{proof}
       Since for $x \in \R^2\setminus\{0\}$ we have
       \[ x'\Sigma x = \frac1{12} \sum_{k \in \Z^2} |f_k|^2 \langle k, x\rangle^2 \geq 0, \]
       the matrix $\Sigma$ is positive semidefinite. Hence, $\det(\Sigma)=0$ iff there is an $x \in \R^2\setminus\{0\}$ s.t. $x'\Sigma x = 0$. This is the case iff
       \begin{equation}
          \label{eq:zerononzero}
          |f_k|^2 \neq 0 \text{ implies } \langle k, x\rangle^2 = 0 \quad \text{for all } k \in \Z^2.
       \end{equation}
       
       If this implication holds, we have for all $y \in \R^2$ and $r \in \R$ that
       \[ f(y+rx) = \sum_{k \in \Z^2} f_k e^{2\pi i\langle k, y+rx\rangle}
          = \sum_{k \in \Z^2} f_k e^{2\pi i\langle k, y\rangle} e^{2\pi i r\langle k, x\rangle}
          = \sum_{k \in \Z^2} f_k e^{2\pi i\langle k, y\rangle}
          = f(y), \]
       i.e. (\ref{eq:fconstant}). If, on the other hand, (\ref{eq:fconstant}) holds, then the two functions $f$ and $f^{rx}(\, \cdot \, ) := f(\, \cdot \, + rx)$ are identical. Subsequently, their respective Fourier coefficients $f_k$ and $f^{rx}_k = e^{2\pi i r\langle k, x\rangle} f_k$ are also the same, i.e. (\ref{eq:zerononzero}) holds.
\end{proof}

\singlespacing

\ifthenelse{\equal{\user}{alex}}{
  \bibliographystyle{abbrv}
  \bibliography{shape,drift} 
  }{}
\ifthenelse{\equal{\user}{stephan}}{
  \bibliographystyle{../../BIB/abbrv}
  \bibliography{../../BIB/shape,drift} 

\begin{thebibliography}{10}

\bibitem{Allassonniere2007}
S.~Allassonni\`ere, Y.~Amit, and A.~Trouv\'e.
\newblock {Towards a Coherent Statistical Framework for Dense Deformable
  Template Estimation}.
\newblock {\em Journal of the Royal Statistical Society, Series B},
  69(1):3--29, 2007.

\bibitem{Antoniadis2006}
A.~Antoniadis and J.~Bigot.
\newblock {Poisson Inverse Problems}.
\newblock {\em Annals of Statistics}, 34(5):2132--2158, 2006.

\bibitem{Aspelmeier2014}
T.~Aspelmeier, A.~Egner, and A.~Munk.
\newblock {Modern Statistical Challenges in High Resolution Fluorescence
  Microscopy}.
\newblock {\em {Annual Review of Statistics and its Application}}, 2, 2014.
\newblock 80pp. To appear.

\bibitem{Babcock2012}
H.~Babcock, Y.~M. Sigal, and X.~Zhuang.
\newblock {A High-density 3D Localization Algorithm for Stochastic Optical
  Reconstruction Microscopy}.
\newblock {\em {Optical Nanoscopy}}, 1(6), 2012.

\bibitem{Berning2012}
S.~Berning, K.~I. Willig, H.~Steffens, P.~Dibaj, and S.~W. Hell.
\newblock {Nanoscopy in a Living Mouse Brain}.
\newblock {\em {Science}}, 335(6068):551, 2012.

\bibitem{Betzig2006}
E.~Betzig, G.~H. Patterson, R.~Sougrat, O.~W. Lindwasser, S.~Olenych, J.~S.
  Bonifacino, M.~W. Davidson, J.~Lippincott-Schwartz, and H.~F. Hess.
\newblock {Imaging Intracellular Fluorescent Proteins at Nanometer Resolution}.
\newblock {\em {Science}}, 313(5793):1642--1645, 2006.

\bibitem{Bickel1998}
P.~J. Bickel, C.~A.~J. Klaassen, Y.~Ritov, and J.~A. Wellner.
\newblock {\em {Efficient and Adaptive Estimation for Semiparametric Models}}.
\newblock {Springer}, 1998.

\bibitem{Bigot2013}
J.~Bigot, S.~Gadat, T.~Klein, and C.~Marteau.
\newblock {Intensity Estimation of Non-homogeneous Poisson Processes from
  Shifted Trajectories}.
\newblock {\em {Electronic Journal of Statistics}}, 7(2013):881--931, 2013.

\bibitem{BigotGamboaVimond2009}
J.~Bigot, F.~Gamboa, and M.~Vimond.
\newblock {Estimation of Translation, Rotation, and Scaling Between Noisy
  Images Using the Fourier-Mellin Transform}.
\newblock {\em {SIAM Journal on Imaging Sciences}}, 2(2):614--645, 2009.

\bibitem{BissantzClaeskensHolzmannMunk2009}
N.~Bissantz, G.~Claeskens, H.~Holzmann, and A.~Munk.
\newblock {Testing for Lack of Fit in Inverse Regression — With Applications
  to Biophotonic Imaging}.
\newblock {\em Journal of the Royal Statistical Society, Series B},
  71(1):25--48, 2009.

\bibitem{BrownLevine2007}
L.~D. Brown and M.~Levine.
\newblock {Variance Estimation in Nonparametric Regression via the Difference
  Sequence Method}.
\newblock {\em Annals of Statistics}, 35(5):2219--2232, 2007.

\bibitem{Bruhn2005}
A.~Bruhn, J.~Weickert, and C.~Schn\"orr.
\newblock {Lucas/Kanade Meets Horn/Schunck: Combining Local and Global Optic
  Flow Methods}.
\newblock {\em International Journal of Computer Vision}, 61(3):211--231, 2005.

\bibitem{Cavalier2002}
L.~Cavalier and J.~Y. Koo.
\newblock {Poisson Intensity Estimation for Tomographic Data Using a Wavelet
  Shrinkage Approach}.
\newblock {\em {IEEE Transactions on Information Theory}}, 48(10):2794--2802,
  2002.

\bibitem{Chen2010}
X.~Chen, J.~Yang, Q.~Wu, and J.~Zhao.
\newblock {Motion Blur Detection Based on Lowest Directional High-frequency
  Energy}.
\newblock In {\em {Proceedings of 2010 IEEE 17th International Conference on
  Image Processing}}, pages 2533--2536, 2010.

\bibitem{Cox2012}
S.~Cox, E.~Rosten, J.~Monypenny, T.~Jovanovic-Talisman, D.~T. Burnette,
  J.~Lippincott-Schwarz, G.~E. Jones, and R.~Heintzmann.
\newblock {Bayesian Localization Microscopy Reveals Nanoscale Podosome
  Dynamics}.
\newblock {\em {Nature Methods}}, 9(2):195--200, 2012.

\bibitem{Cuzol2007}
A.~Cuzol, P.~Hellier, and E.~Memin.
\newblock {A Low Dimensional Fluid Motion Estimator}.
\newblock {\em International Journal of Computer Vision}, 75(3):329--349, 2007.

\bibitem{Deschout2014}
H.~Deschout, F.~C. Zanacchi, M.~Mlodzianoski, A.~Diaspro, J.~Bewersdorf, S.~T.
  Hess, and K.~Braeckmans.
\newblock {Precisely and Accurately Localizing Single Emitters in Fluorescence
  Microscopy}.
\newblock {\em {Nature Methods}}, 11(3):253--266, 2014.

\bibitem{Dette1998}
H.~Dette and A.~Munk.
\newblock {Testing Heteroscedasticity in Nonparametric Regression}.
\newblock {\em Journal of the Royal Statistical Society, Series B},
  60(4):693--708, 1998.

\bibitem{Egner2007}
A.~Egner, C.~Geisler, C.~von Middendorff, H.~Bock, D.~Wenzel, R.~Medda,
  M.~Andresen, A.~Stiel, S.~Jakobs, C.~Eggeling, A.~Schoenle, and S.~W. Hell.
\newblock {Fluorescence Nanoscopy in Whole Cells by Asynchronous Localization
  of Photoswitching Emitters}.
\newblock {\em {Biophysical Journal}}, 93(9):3285--3290, 2007.

\bibitem{Evans1998}
L.~C. Evans.
\newblock {\em {Partial Differential Equations (Graduate Studies in Mathematics
  vol 19)(Providence, {RI}: American Mathematical Society)}}.
\newblock {Oxford University Press}, 1998.

\bibitem{Fleet2006}
D.~J. Fleet and Y.~Weiss.
\newblock {\em {Optical Flow Estimation}}, pages 237--257.
\newblock Springer, 2006.

\bibitem{Foroosh2002}
H.~Foroosh, J.~Zerubia, and M.~Berthod.
\newblock {Extension of Phase Correlation to Subpixel Registration}.
\newblock {\em {IEEE Transactions on Image Processing}}, 11(3):188--200, 2002.

\bibitem{Frick2013}
K.~Frick, M.~Marnitz, and A.~Munk.
\newblock {Statistical Multiresolution Estimation for Variational Imaging: With
  an Application in {P}oisson-Biophotonics}.
\newblock {\em {Journal of Mathematical Imaging and Vision}}, 46(3):370--387,
  2013.

\bibitem{GamboaLoubesMaza2007}
F.~Gamboa, J.-M. Loubes, and E.~Maza.
\newblock {Semi-parametric Estimation of Shifts}.
\newblock {\em {Electronic Journal of Statistics}}, 1:616--640, 2007.

\bibitem{Geisler2012}
C.~Geisler, T.~Hotz, A.~Schoenle, S.~W. Hell, A.~Munk, and A.~Egner.
\newblock {Drift Estimation for Single Marker Switching Based Imaging Schemes}.
\newblock {\em {Optics Express}}, 20(7):7274--7289, 2012.

\bibitem{Geisler2007}
C.~Geisler, A.~Schoenle, C.~von Middendorff, H.~Bock, C.~Eggeling, A.~Egner,
  and S.~W. Hell.
\newblock {Resolution of ${\lambda/10}$ in Fluorescence Microscopy Using Fast
  Single Molecule Photo-switching}.
\newblock {\em {Applied Physics A}}, 88(2):223--226, 2007.

\bibitem{Gonzalez2002}
{Gonzalez, R.C. and Woods, R.E.}
\newblock {\em {Digital Image Processing}}.
\newblock {Prentice Hall}, 2 edition, 2002.

\bibitem{Gow}
J.~C. Gower.
\newblock Generalized {P}rocrustes analysis.
\newblock {\em Psychometrika}, 40(1):33--51, 1975.

\bibitem{Gustafsson2005m}
M.~G.~L. Gustafsson.
\newblock {Nonlinear Structured-illumination Microscopy: Wide-field
  Fluorescence Imaging with Theoretically Unlimited Resolution}.
\newblock {\em {Proceedings of the National Academy of Sciences of the United
  States of America}}, 102(37):13081--13086, 2005.

\bibitem{Hafi2014}
N.~Hafi, M.~Grunwald, L.~S. van~den Heuvel, T.~Aspelmeier, J.-H. Chen,
  M.~Zagrebelsky, O.~M. Sch\"utte, C.~Steinem, M.~Korte, A.~Munk, and P.~J.
  Walla.
\newblock {Fluorescence Nanoscopy by Polarization Modulation and Polarization
  Angle Narrowing}.
\newblock {\em {Nature Methods}}, 11(5):579--584, 2014.

\bibitem{HallPittelkow}
P.~Hall and Y.~E. Pittelkow.
\newblock {Simultaneous Bootstrap Confidence Bands in Regression}.
\newblock {\em {Journal of Statistical Computation and Simulation}},
  37(1--2):99--113, 1990.

\bibitem{Heintzmann2002}
R.~Heintzmann, T.~M. Jovin, and C.~Cremer.
\newblock {Saturated Patterned Excitation Microscopy--A Concept for Optical
  Resolution Improvement}.
\newblock {\em {Journal of the Optical Society of America A: Optics, Image
  Science, and Vision}}, 19(8):1599--1609, 2002.

\bibitem{Hell2003}
S.~W. Hell.
\newblock {Toward Fluorescence Nanoscopy}.
\newblock {\em {Natural Biotechnology}}, 21(11):1347--1355, 2003.

\bibitem{Hell2007}
S.~W. Hell.
\newblock {Far-field Optical Nanoscopy}.
\newblock {\em {Science}}, 316(5828):1153--1158, 2007.

\bibitem{Hell2009b}
S.~W. Hell.
\newblock {Microscopy and its Focal Switch}.
\newblock {\em {Nature Methods}}, 6(1):24--32, 2009.

\bibitem{Hell1994}
S.~W. Hell and J.~Wichmann.
\newblock {Breaking the Diffraction Resolution Limit by Stimulated Emission:
  Stimulated-emission-depletion Fluorescence Microscopy}.
\newblock {\em {Optics Letters}}, 19(11):780--782, 1994.

\bibitem{Hess2006}
S.~T. Hess, T.~P.~K. Girirajan, and M.~D. Mason.
\newblock {Ultra-high Resolution Imaging by Fluorescence Photoactivation
  Localization Microscopy}.
\newblock {\em {Biophysical Journal}}, 91(11):4258--4272, 2006.

\bibitem{Hofmann2005}
M.~Hofmann, C.~Eggeling, S.~Jakobs, and S.~W. Hell.
\newblock {Breaking the Diffraction Barrier in Fluorescence Microscopy at Low
  Light Intensities by Using Reversibly Photoswitchable Proteins}.
\newblock {\em {Proceedings of the National Academy of Sciences of the United
  States of America}}, 102(49):17565--17569, 2005.

\bibitem{Holden2011}
S.~J. Holden, S.~Uphoff, and A.~N. Kapanides.
\newblock {DAOSTORM: An Algorithm for High-density Super-resolution
  Microscopy}.
\newblock {\em {Nature Methods}}, 8(4):279--280, 2011.

\bibitem{Huang2013}
F.~Huang, T.~M.~P. Hartwich, F.~E. Rivera-Molina, Y.~Lin, W.~C. Duim, J.~J.
  Long, P.~D. Uchil, J.~R. Myers, M.~A. Baird, W.~Mothes, M.~W. Davidson,
  D.~Toomre, and J.~Bewersdorf.
\newblock {Video-rate Nanoscopy Using sCMOS Camera-specific Single-molecule
  Localization Algorithms}.
\newblock {\em {Nature Methods}}, 10(7):653--658, 2013.

\bibitem{Huang2005}
J.-Z. Huang, T.-N. Tan, L.~Ma, and Y.-H. Wang.
\newblock {Phase Correlation Based Iris Image Registration Model}.
\newblock {\em {Journal of Computer Science and Technology}}, 20(3):419--425,
  2005.

\bibitem{Jones2011}
S.~A. Jones, S.-H. Shim, J.~He, and X.~Zhuang.
\newblock {Fast, Three-dimensional Super-resolution Imaging of Live Cells}.
\newblock {\em {Nature Methods}}, 8(6):499--508, 2011.

\bibitem{Klar2000}
T.~A. Klar, S.~Jakobs, M.~Dyba, A.~Egner, and S.~W. Hell.
\newblock {Fluorescence Microscopy with Diffraction Resolution Barrier Broken
  by Stimulated Emission}.
\newblock {\em {Proceedings of the National Academy of Sciences of the United
  States of America}}, 97(15):8206--8210, 2000.

\bibitem{Li2013}
H.~Li, M.~Haltmeier, S.~Zhang, J.~Frahm, and A.~Munk.
\newblock {Aggregated Motion Estimation for Image Reconstruction in Real-Time
  MRI}.
\newblock {\em {Magnetic Resonance in Medicine, to appear}}, 2013.

\bibitem{Liu1988}
R.~Y. Liu.
\newblock {Bootstrap Procedures under some Non-i.i.d. Models}.
\newblock {\em Annals of Statistics}, 16(4):1696--1708, 1988.

\bibitem{Mammen1993}
E.~Mammen.
\newblock {Bootstrap and Wild Bootstrap for High Dimensional Linear Models}.
\newblock {\em Annals of Statistics}, 21(1):255--285, 1993.

\bibitem{Munk2005}
A.~Munk, N.~Bissantz, T.~Wagner, and G.~Freitag.
\newblock {On Difference-based Variance Estimation in Nonparametric Regression
  when the Covariate is High Dimensional}.
\newblock {\em Journal of the Royal Statistical Society, Series B},
  67(1):19--41, 2005.

\bibitem{Nowak2000}
R.~D. Nowak and E.~D. Kolaczyk.
\newblock {A Statistical Multiscale Framework for Poisson Inverse Problems}.
\newblock {\em {IEEE Transactions in Information Theory}}, 46(5):1811--1825,
  2000.

\bibitem{Papenberg2006}
N.~Papenberg, A.~Bruhn, T.~Brox, S.~Didas, and J.~Weickert.
\newblock {Highly Accurate Optic Flow Computation with Theoretically Justified
  Warping}.
\newblock {\em International Journal of Computer Vision}, 67(2):141--158, 2006.

\bibitem{Quan2011}
T.~Quan, H.~Zhu, X.~Liu, Y.~Liu, J.~Ding, S.~Zeng, and Z.-L. Huang.
\newblock {High-density Localization of Active Molecules Using Structured
  Sparse Model and Baysian Information Criterion}.
\newblock {\em {Optics Express}}, 19(18):16963--16974, 2011.

\bibitem{reddy}
B.~Reddy and B.~Chatterji.
\newblock {An FFT-based Technique for Translation, Rotation, and
  Scale-invariant Image Registration}.
\newblock {\em {IEEE Transactions on Image Processing}}, {5}({8}):{1266--1271},
  {1996}.

\bibitem{Rust2006}
M.~J. Rust, M.~Bates, and X.~W. Zhuang.
\newblock {Sub-diffraction-limit Imaging by Stochastic Optical Reconstruction
  Microscopy (STORM)}.
\newblock {\em {Nature Methods}}, 3(10):793--795, 2006.

\bibitem{Schick2008}
A.~Schick and W.~Wefelmeyer.
\newblock {Some Developments in Semiparametric Models}.
\newblock {\em {Journal of Statistical Theory and Practice}}, 2(3):475--491,
  2008.

\bibitem{Schmidt2008}
R.~Schmidt, C.~A. Wurm, S.~Jakobs, J.~Engelhardt, A.~Egner, and S.~W. Hell.
\newblock {Spherical Nanosized Focal Spot Unravels the Interior of Cells}.
\newblock {\em {Nature Methods}}, 5(6):539--544, 2008.

\bibitem{Sen1993}
P.~K. Sen and J.~M. Singer.
\newblock {\em {Large Sample Methods in Statistics}}.
\newblock Chapman \& Hall, 1993.

\bibitem{Silverman1990}
B.~W. Silverman, M.~C. Jones, J.~D. Wilson, and D.~W. Nychka.
\newblock {A Smoothed EM Approach to Indirect Estimation Problems, with
  Particular Reference to Stereology and Emission Tomography}.
\newblock {\em Journal of the Royal Statistical Society, Series B},
  52(2):271--324, 1990.

\bibitem{Thompson2002}
R.~Thompson, D.~Larson, and W.~Webb.
\newblock {Precise Nanometer Localization Analysis for Individual Fluorescent
  Probes}.
\newblock {\em {Biophysical Journal}}, 82(5):2775–--2783, 2002.

\bibitem{Vaart2000}
A.~W. van~der Vaart.
\newblock {\em {Asymptotic Statistics}}.
\newblock {Cambridge Series in Statistical and Probabilistic Mathematics}.
  {Cambridge University Press}, 2000.

\bibitem{Vardi1985}
Y.~Vardi, L.~A. Shepp, and L.~Kaufman.
\newblock {A Statistical Model for Positron Emission Tomography}.
\newblock {\em Journal of the American Statistical Association}, 80(389):8--20,
  1985.

\bibitem{Weickert2001}
J.~Weickert and C.~Schn\"orr.
\newblock {A Theoretical Framework for Convex Regularizers in PDE-Based
  Computation of Image Motion}.
\newblock {\em International Journal of Computer Vision}, 45(3):245--264, 2001.

\bibitem{Westphal2008}
V.~Westphal, S.~O. Rizzoli, M.~A. Lauterbach, D.~Kamin, R.~Jahn, and S.~W.
  Hell.
\newblock {Video-rate Far-field Optical Nanoscopy Dissects Synaptic Vesicle
  Movement}.
\newblock {\em {Science}}, 320(5873):246--249, 2008.

\bibitem{Wu1986}
C.~F.~J. Wu.
\newblock {Jackknife, Bootstrap, and Other Resampling Methods in Regression
  Analysis (with Discussion)}.
\newblock {\em Annals of Statistics}, 14(4):1261--1295, 1986.

\bibitem{Xu2013}
W.~Xu, J.~Mulligan, D.~Xu, and X.~Chen.
\newblock {Detecting and Classifying Blurred Image Regions}.
\newblock In {\em {Proceedings of 2013 IEEE International Conference on
  Multimedia and Expo (ICME)}}, pages 1--6, 2013.

\bibitem{Zhang2008}
B.~Zhang, J.~M. Fadili, and J.~L. Starck.
\newblock {Wavelets, Ridgelets, and Curvelets for Poisson Noise Removal}.
\newblock {\em {IEEE Transactions on Image Processing}}, 17(7):1093--1108,
  2008.

\bibitem{Zhu2012}
L.~Zhu, W.~Zhang, D.~Elnatan, and B.~Huang.
\newblock {Faster STORM Using Compressed Sensing}.
\newblock {\em {Nature Methods}}, 9(7):721--726, 2012.

\end{thebibliography}
  }{}

\end{document}